\documentclass{article}
\usepackage[utf8]{inputenc}
\usepackage{amsfonts}
\usepackage{amsmath}
\usepackage{authblk}

\usepackage{amsmath, amssymb, amsthm, mathtools, mathrsfs, bbm}
\usepackage[margin=1in]{geometry}

\newcommand{\R}{\mathbb{R}}

\newcommand{\C}{\mathbb{C}}

\newcommand{\CC}{\mathcal{C}}

\theoremstyle{plain}
\newtheorem{theorem}{Theorem}[section]

\newtheorem{lemma}[theorem]{Lemma}

\theoremstyle{definition}

\newtheorem{remark}[theorem]{Remark}

\numberwithin{equation}{section}

\def\ssgbegin{\begin{itemize}}

\def\ssgend{\end{itemize}}
\newcommand{\ignore}[1]{}

\def\IR{\mathbb{R}}

\def\IZ{\mathbb{Z}}
\def\E{{\mathbb E}}
\def\tr{\mbox{tr}\;}

\newcommand{\be}{\begin{equation}}
\newcommand{\ee}{\end{equation}}
\newcommand{\ba}{\begin{array}}
\newcommand{\ea}{\end{array}}
\newcommand{\bea}{\begin{eqnarray}}
\newcommand{\eea}{\end{eqnarray}}
\usepackage{pstricks,pst-node,pst-text,pst-3d}
\definecolor{lightgray}{gray}{0.95}
\usepackage{color,xcolor}
\definecolor{keywordcolor}{rgb}{0.5, 0.1, 0.1}  
\definecolor{tacticcolor}{rgb}{0.1, 0.2, 0.4}   
\definecolor{commentcolor}{rgb}{0.4, 0.4, 0.4}  
\definecolor{symbolcolor}{rgb}{0.0, 0.1, 0.6}   
\definecolor{sortcolor}{rgb}{0.1, 0.5, 0.1}     
\definecolor{attributecolor}{rgb}{0.7, 0.1, 0.1}
\usepackage[colorlinks=true,linkcolor=blue,urlcolor=blue,citecolor=blue]{hyperref}
\usepackage{upgreek}
\usepackage{listings}

\usepackage{xcolor}

\title{Formalization of QFT} 
\date{March 2026}

\author[1]{Michael R. Douglas \thanks{mdouglas@cmsa.fas.harvard.edu}}
\author[2]{Sarah Hoback \thanks{ sarahhoback@g.harvard.edu,  now at Physical Superintelligence PBC: sarah@psi.inc}}
\author[2]{Anna Mei\thanks{amei@g.harvard.edu}}
\author[3]{Ron Nissim \thanks{rnissim@mit.edu}}

\affil[1]{Center Of Mathematical Sciences And Applications, Harvard University, Cambridge, MA 02138 USA }
\affil[2]{Jefferson Physical Laboratory, Harvard University, Cambridge, MA 02138 USA}
\affil[3]{Simons Building, MIT, Cambridge, MA 02139 USA}

\begin{document}
\maketitle{}

\begin{abstract}
A foundational result in constructive quantum field theory is the construction of the free bosonic quantum field theory in four-dimensional Euclidean spacetime and the proof that it satisfies the Glimm–Jaffe axioms, a variant of the Osterwalder–Schrader axioms. We present a formalization of this result in the Lean 4 interactive theorem prover. The project is intended as a proof of concept that extended arguments in mathematical physics can be translated into machine-checked proofs using existing AI tools. We begin by introducing interactive theorem proving and constructive quantum field theory, then describe our formalization and the design decisions that shaped it. We also explain the methods we used, including coding assistants, and conclude by considering how AI assisted formalization may influence the future of theoretical physics.

Our original release assumed three results, Minlos’ theorem, the nuclear property of Schwartz space, and Goursat's theorem. 
In subsequent releases from our group and from contributors from the Lean community,
these assumptions have been proven (or avoided),  so that the OS/GJ axioms are now proven using only Lean and its library Mathlib.
\end{abstract}
\vfill\eject

\tableofcontents

\newpage

\section{Introduction}
\label{s:intro}

AI is starting to have a major impact on science.  This impact is happening in many ways, and one of them is its use in interactive theorem proving (ITP), a technology for proving and checking mathematical theorems expressed in a formal language, via interactive theorem proving and AI coding assistants. There is rapid and recent development that demonstrates the practical maturity of machine-proven and verified mathematics at scale, including Lean-verified solutions achieving IMO 2025 gold-medal-level performance, the release of end-to-end formal (Lean 4) solution corpora for Putnam 2025 by autonomous theorem-proving systems, and a rapid expansion of autoformalization methods and benchmarks aimed at translating informal statements/proofs into verifiable Lean code \cite{harmonic_aristotle_imo2025_2025,harmonic_imo2025_repo_2025}. Many mathematicians are optimistic about the potential impact ITP will have on research level mathematics, though its influence elsewhere has been far less robust. 

We believe that ITP will also have great value for theoretical physics. To demonstrate that it is now practical, we have formalized the construction of free massive bosonic $d=4$ Euclidean quantum field theory (QFT) in Lean, following the treatment of Glimm and Jaffe
\cite{glimm2012quantum}. Our code is in the GitHub repo {\tt mrdouglasny/OSforGFF} together with its dependencies. In this preprint we explain our methods and many details.

As formalization is not a widely known concept in the theoretical physics community, we will begin by explaining what it is and why physicists should be interested in it.  We then explain our goals for this project, share some important foundational and technical results, and most importantly we discuss the definitions of QFT within various axiomatic systems and the reconstruction theorems which relate them.  We also outline more ambitious goals which we believe will be possible with strong enough formalization technology. The authors of this paper have felt the formalization landscape change beneath their feet during the course of this formalization project. The strength of the formalization technology we used, the publicly available Claude Code, GPT Codex and Gemini models, significantly improved over the span of 9 months. It is not out of the question that ambitious projects outlined in the conclusions will come into range in 2026.

At present, formalization is based on mathematical foundations, so to formalize
QFT we must work with rigorous mathematical definitions.
In \S \ref{s:math} we give a brief introduction to constructive QFT, both outlining
a few of the major approaches and then diving into the details we will need.  One
aspect of this is that the same physical ideas can be expressed in several mathematically different ways, all leading to comparable final results.  These distinctions become even more significant in a formal setting.  Concrete choices must be made, and these choices are strongly influenced by practical considerations, such as which definitions and theorems are already available in Lean libraries and how difficult the resulting proofs are to implement. As a result, we occasionally adopt constructions that are not the most conventional from a physics perspective, but which fit more smoothly into the existing formal framework. 
In \S \ref{s:details}, we explain the main choices we made, the reasons for those choices, and give an overview of the formal proofs.

In \S \ref{s:methods} we discuss our methods -- how we began, workflow, and tools for managing this medium size collaborative software project. Some of this is already outdated as LLM and coding assistant technology is evolving extremely rapidly.
In early 2025 coding assistants were not able to write Lean code at the required level.  Perhaps by 2027 we will just give them the informal math statement and they will do all the work.  
In \S \ref{s:going-forward} we give advice and suggestions from today's (early 2026) standpoint.  
In addition to introducing formal theorists to ITP and encouraging them to consider how it can be used in their research, 
we intend this to serve as a historical record of the current progress in AI-assisted formal verification.

\subsection{What is formalization and why?}

Interactive theorem proving (ITP) is a method for developing and checking mathematical proofs 
by expressing them in a formal language, a language with precise logical semantics
which can be checked by a computer.
The process of translating a piece of ``informal'' mathematics into a formal language is called formalization, and is  
presently done through interactive collaboration between a user and a proof assistant \cite{harrison_history_itp_2014}.
There is a related concept of automated theorem proving (ATP), which does not require human collaboration, while automated translation
to formal mathematics is called autoformalization \cite{Szegedy2020APP}.

ITP has been developed for many decades, and several proof assistants are well suited to expressing mathematics, including Coq/Rocq \cite{rocq_website}, Isabelle \cite{paulson_isabelle_1994}, and more recently Lean \cite{avigad_tp_lean4}.
Lean is the focus of much current community effort \cite{avigad_tp_lean4, math_in_lean, demoura_ullrich_lean4}
and has a sizable library of standard definitions and results called Mathlib \cite{DBLP:journals/corr/abs-1910-09336}.
Some recent introductions to Lean are \cite{avigad_tp_lean4, math_in_lean}.
To give a flavor of Lean, Figure~\ref{fig:lean-proof-snippet} shows a short Lean proof excerpt (adapted from \cite{avigad_tp_lean4}). We will save more detailed discussion for our actual topics, but at first glance one sees programming code not too different from familiar languages such as Python and Mathematica.  This is correct, with the crucial difference that the operations, both predefined and user defined, are all constrained to perform sound logical deductions.  The elementary steps include applying rules of logical inference, bringing in axioms, substituting variables and so on.  By the use of subroutine definitions and the sophisticated dependent type system,
logical arguments composed of many, many elementary steps can be expressed concisely.  If suitably written,
Lean code can be executed in the same way as other programming languages.  But one also has the option to use
non-computable logical arguments, such as the principle of the excluded middle (for any $X$, one of $X$ or $\mbox{not}\, X$ 
must be true).  Thus Lean and the other advanced ITP systems are capable of expressing any mathematical statement
and of verifying proofs which use the full range of mathematical technique.

\begin{figure}[t]\label{fig:lean-proof-snippet}
  \centering
  \includegraphics[width=\linewidth]{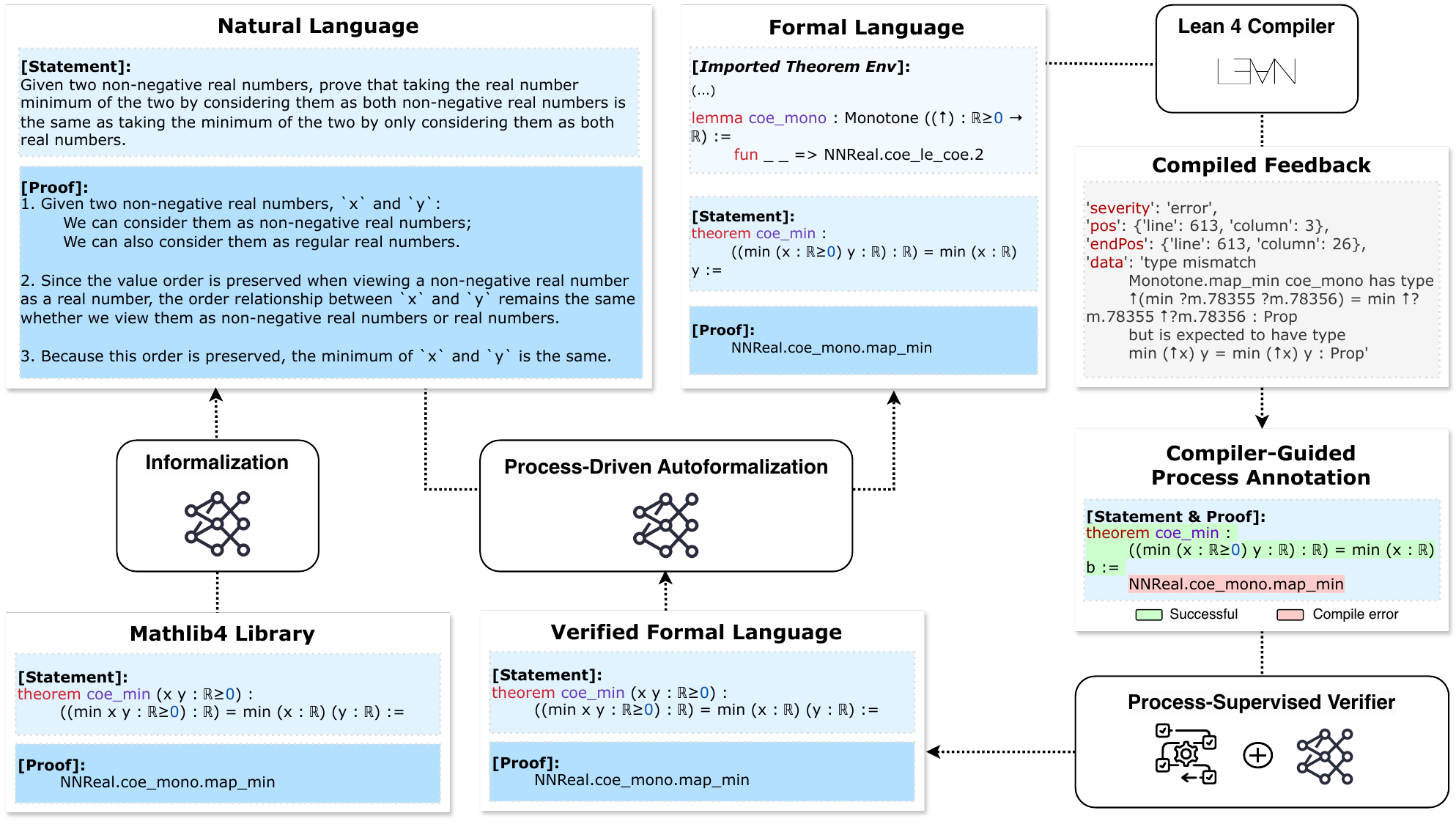}
  \caption{\textbf{PDA trained on FORML4 ( image from \cite{lu_process_driven_2024}).}
  PDA is aimed at \emph{statement} autoformalization rather than proof translation; proof steps are included to let the Lean compiler provide process-level feedback when jointly compiling the statement and proof steps. In the depicted example, the statement parses but a proof-step error reveals an incorrect autoformalization, guiding refinement  \cite{lu_process_driven_2024}.)}
  \label{fig:pda-overview}
\end{figure}

Clearly a technology which speeds up proof and is able to verify it in complete rigor would be extremely valuable throughout the mathematical sciences.  Mathematics journals are notorious for their lengthy review times, and this is largely because checking all of the details of a proof is laborious and time-consuming.  Even then there is no guarantee; there are famous cases of proofs which were published and generally accepted, only for flaws to be found later.  In addition to verification, another major benefit of ITP is that it facilitates collaboration.  Often the main bottleneck in assembling a large collaboration is the need to vet the contributors and verify their contributions; this difficulty could be largely removed. For these reasons, mathematicians are excited about the potential of ITP \cite{tao_pfr_blueprint_2023}.

How excited should theoretical physicists be about ITP?  
One significant barrier is that formalization, at least in its current state, requires all arguments to be stated with precise definitions and complete proofs, whereas much of the physics literature deliberately operates at a looser level of rigor. 
We will further discuss advantages of this flexible style below \S\ref{limitsofformalizatin},
but one must admit that
it introduces substantial ambiguity into scientific communication.  An example in the context of QFT is the definition of a path integral,
which is very nontrivial and involves many choices which are often left implicit, so that claims based on it must be carefully re-evaluated
in each new application.  By contrast, mathematical rigor and even more so
formal rigor force every mathematical object to be explicitly defined and every assumption to be stated. In doing so, one must articulate clearly the mathematical content of a theory and the precise hypotheses under which the arguments are valid, often revealing hidden assumptions that were taken for granted.  The resulting framework, while more demanding to construct, provides a shared and unambiguous foundation that can facilitate collaboration and long-term reuse of results, even in very different contexts.

Whatever one thinks about the advantages and disadvantages of rigor,
as formalization exists now it must be based on rigorous mathematics.
And as any physicist who has tried to read mathematical physics
knows, the structure of the discussions and proofs is  rather different from most of
theoretical physics, with a heavy reliance on functional analysis.  
This added precision can be advantageous, but not always.
Is it really true that we need rigor for formalization?  Could formalization be based on some other foundations more analogous to standard physical reasoning?
We will return to these questions in the conclusions.
But to use ITP now, we must formalize a piece of mathematical physics which is rigorously stated and proven.  An important example is the problem of rigorously defining QFT and proving that examples exist, often called ``constructing'' the QFT. This has been much studied and there are many results to draw on.   We chose a relatively simple but very foundational project to find out whether formalization is now practical.
And we learned quite a lot doing it, about the mathematical physics, about theorem proving and about AI, which we report here.

Formalization could be of great value to both math and physics. Yet the first question is, is it practical to formalize research
level mathematical physics such as QFT at all?
Why was this not already done?

\subsection{ITP in the age of AI}

As we write this in early 2026, ITP has not yet had a major impact on how mathematics research is done. One reason arises from the observation that much of mathematics builds on a deep structure of definitions based on other definitions, which takes years of study to gain even a basic familiarity.  This is also true of much of theoretical physics.  Now, to formalize a piece of mathematics, its prerequisite definitions and background theory must already be available in the library to build on; in practice one may initially omit some proofs by treating results as unchecked ``axioms'' while continuing development.
But even with this relaxed standard, only a tiny fraction of mathematics has been fully formalized so far \cite{massot_why_formalize_2024}.
Until substantially more of the mathematical corpus is formalized, the scope of ITP will remain limited. Moreover, even after researchers formalize important theorems of value to the broader community, the process of integrating their work into a foundational library is slow.  This is for good reasons: there  are important standards upheld by Mathlib for new code to be integrated into the library, namely it should be humanly readable and follow community-developed principles of programming and mathematical organization.
Still, this seems to be a serious bottleneck to wider adoption of the technology.

Furthermore, ITP has a steep learning curve for new users, taking months of study to gain even moderate proficiency
\cite{thiemann_structuring_definitions_2025,zhang_learning_rules_2024}. Formalization is laborious, even for experts. This contrasts with tools like computer algebra which can be profitably used after a few days of study.  Historically, it has taken around a week to formalize a single page of a textbook \cite{wiedijk2007formalization}. 
The Lean Mathlib library \cite{math_in_lean} 
is growing at about 300K lines/year; taking into account that a formal proof is about 5 times as long as an informal proof this corresponds to the entire community of several hundred active contributors formalizing about one textbook every 2 months.  At present Mathlib has fairly good coverage of the undergraduate math curriculum, but at this rate it will take many years to formalize the prerequisites for research level topics in differential geometry, PDE, Lie algebras and groups, and so on.
Another library PhysLean \cite{physlean_website} has been started with a focus on physics formalization, but these gaps remain.
Thus until recently, the answer to the question we raised on the practicality of formalizing constructive QFT was,  not really.

However, this situation is rapidly improving. Automated theorem proving technologies such as Isabelle's ``Sledgehammer'' \cite{blanchette_bohme_paulson_sledgehammer_2011}, and newer protocols for agentic interaction with Lean via the Model Context Protocol (MCP) \cite{lean_lsp_mcp_github,axprover_2025,mcp_spec_2025} have made steady progress.
Of potentially greater import, one can apply modern AI to interactive theorem proving; this direction, often referred to as AI for theorem proving (AITP), has been actively pursued since at least 2015 \cite{aitp_announcement_2015,aitp_site}.
Many experts believe this direction is starting to take off; for example, DARPA's expMath initiative aims to radically (``exponentially'') accelerate the rate of progress in pure mathematics by developing AI systems for decomposition and autoformalization \cite{darpa_expmath_program}.
At this rate, even the present Lean community could formalize 600 textbooks/year and rapidly bring much of math and physics into the scope of ITP.
And if the techniques become easy to learn, one can foresee widespread adoption.

There are several AITP paradigms. One is to view proof search as a game-like sequential decision process (often likened to solitaire), and to apply game-playing methods such as reinforcement learning (RL) and Monte Carlo search \cite{huang_gamepad_2018,bansal_holist_2019}.
The power of RL was famously demonstrated by AlphaGo and AlphaZero \cite{silver_alphago_2016,silver_alphazero_2018}, and
AlphaProof is a recent example in formal mathematics, combining RL with formal verification inside Lean \cite{hubert_alphaproof_2025}.

Another paradigm is to view formalization as a form of language translation and to automate it, autoformalization, using AI methods for language translation \cite{lu_process_driven_2024,poiroux_reliable_2025}.
This naturally connects with the remarkable recent development of large language models (LLMs) \cite{brown_gpt3_2020,openai_gpt4_2023}, which have revolutionized computational linguistics, and are now revolutionizing computer programming and many of the other tasks which lie at the heart of computing.  LLMs have remarkable problem solving and reasoning abilities and are advancing rapidly. There are even projections that ``artificial general intelligence'' (an ill-defined term) will be achieved soon.  Some AI researchers believe that AITP will be a central part of this.  One argument is that the most successful AI technologies are statistical and destined to make errors (at least occasionally), so to make long correct chains of argument they must be coupled to a verification engine such as an ITP system.

The imminent prospect of AI's doing math and science is raising questions which a few years ago would have been dismissed as science fiction.
Already AI is far faster and more prolific in generating
results than any human.   Soon the bottleneck will not be writing the papers, it will be
reading the papers!  How will we understand all of this new science?  A first step towards
understanding is simply to be convinced that the claims and arguments are correct, how
can we convince ourselves that a long paper full of intricate arguments is correct?  Math is already bottlenecked by review time.

Formalization provides a partial answer.  It is a way that the computer
can reliably certify that assumptions $A,B,C$ imply consequences $X,Y,Z$.
It does not completely solve the problem, after all we still need to convince
ourselves that $A,B,C$ were the right assumptions and $X,Y,Z$ answer our questions.
But it would be a great help.

We leave further discussion and speculation for other venues, but the point we draw from this is that ITP technology is advancing quickly and we expect the problems of practical difficulty and insufficient libraries to be largely solved over the next few years (the first author has consistently predicted this by 2030 in many talks as far back as 2019 \cite{douglas:2019}).  We will report below on our experience using the existing tools for theorem proving, and we do see rapid progress: our ability to do ITP with mid-2025 tools was far greater than it was at the beginning of that
year, and the early 2026 tools are noticeably better still.
It seems entirely possible that, over the coming years, AI will so much enhance our abilities (or even achieve superhuman abilities) that the advantages of rigor and formalization will far outweigh the remaining disadvantages.  Thus we believe that yes, theoretical physicists should be excited about ITP.

\subsection{Comparing Styles of Reasoning in Physics and Mathematics}\label{limitsofformalizatin}

A strength that physicists claim over mathematicians is that looser levels of rigor allow them to see far beyond what they can prove \cite{feynman1965character}. 
One joke has it that ``Mathematicians are proving what physicists have known for years.''
A less physics-centric observation is that in exploratory work, before ideas take their final form, often it is better to
maintain flexibility.  As Niels Bohr famously said, ``Never express yourself more clearly than you are able to think'' \cite{pais_genius_science_2000}. 

However, mathematical concepts often arise independently of physics, yet later map onto nature with striking success \cite{Wigner1960UnreasonableEffectiveness}.
Mathematics has obvious implications for physics research, starting with the need for physical theories to be mathematically consistent.  Mathematicians can unblock technical barriers of physics by providing new mathematical machinery and frameworks.   Novel theorems can make complex physical arguments manifestly true.
However, physics is having more and more impact on mathematics. There is a real advantage to making squishier conjectures and using less precise definitions of the objects one is reasoning about. Reasoning under uncertainty about the features and transformations which apply to an object is something physicists thrive on \cite{Atiyah2006InteractionGeometryPhysics}. 

As our understanding of the natural world evolves, our frameworks, definitions, and understanding of relevant key features of a system evolve. Kuhn’s account of scientific development distinguishes periods of “normal science,” in which a community works within a shared paradigm to solve well-posed puzzles, from rare episodes of “scientific revolution,” in which persistent anomalies provoke a crisis and are resolved only by a reorganization of the field’s core conceptual framework. On this view, mature formalisms and tightened standards of admissible argument are not neutral, they are part of what defines a paradigm and stabilizes normal science. This stabilization is productive, but it can also make conceptual innovation harder: anomalies that resist expression in the reigning formalism are more likely to be treated as technical puzzles, deferred, or re-described, rather than as pressure toward a new framework. If so, a comparative advantage of physicists’ informal, heuristic style is that it can generate candidate re-conceptualizations and partial “exemplars” more quickly than fully rigorous approaches can, supplying the kinds of alternatives Kuhn argues are necessary for genuine paradigm change \cite{Kuhn1962StructureScientificRevolutions}.

On a more practical level, one of the greatest difficulties we encountered is that the long algebraic calculations which are common in theoretical physics can
be quite difficult to formalize in Lean.  This is somewhat ironic as when used in the proper contexts, algebraic arguments are rigorous.
We will return to this issue in the conclusions.

\subsection{Brief status report on rigorous QFT with respect to the Mass Gap Problem}
\label{ss:cqft}

This section is written to put this project in the context of Yang-Mills, and to discuss recent developments in progress towards the mass-gap problem. The mass gap problem is the question of whether quantum Yang–Mills theory, the mathematical framework underlying the strong interaction, can be defined rigorously in four dimensions and proved to have a positive gap between the vacuum state and the lightest possible excitation. Informally, even though the classical gauge fields are “massless,” the quantum theory appears to produce only particles with nonzero minimum mass. The Clay Mathematics Institute formulates this as proving that for any compact simple gauge group, a nontrivial quantum Yang–Mills theory exists on ($\mathbb{R}^4$) and has a mass gap ($\Delta>0$). The mass gap problem is of interest to physicists because this gap is tied to the observed short-range, confining behavior of the strong force and to why free gluons are not seen experimentally, so solving the problem would put one of the Standard Model’s most successful pieces on a fully rigorous mathematical foundation rather than the current mix of physical evidence, approximation methods, and lattice simulations \cite{clay_yang_mills_mass_gap}.

In the 1970s there was a great effort among mathematical physicists to make the construction of interacting quantum field theories mathematically rigorous. A successful approach has been to construct a Euclidean version of the theory and then rigorously apply a Wick rotation procedure to construct the original quantum field theory. This approach is all based on a reconstruction theorem carrying out the Wick rotation procedure, provided the Euclidean theory satisfies a set of axioms known as the Osterwalder-Schrader (OS) axioms \cite{OS1972}. Around the same time there was a huge effort to construct Euclidean scalar quantum field theories. This program succeeded in the construction of the Euclidean $P(\Phi)_2$ and $\Phi_3^4$ theory were constructed along with the verification that they satisfy the OS axioms (See i.e. \cite{glimm2012quantum, simon2015p} for an overview of these results). On the other hand, in four dimensions, heuristic calculations indicated that there are no nontrivial scalar field theories, and it was rigorously verified that the $\Phi_d^4$ theory is trivial for $d \geq 5$ \cite{Aiz81}. Much more recently, this was also verified for $d=4$ \cite{AiDu21}.
But despite this, it is widely believed that the four dimensional Yang--Mills theory can be rigorously constructed,
and there are strong physical and computational arguments that it is nontrivial and has the mass gap property.

There has been a recent revival of interest in Yang--Mills theory and constructive quantum field theory in the math community. This is partially due to the introduction of regularity structures and paracontrolled calculus techniques in the theory of singular stochastic partial differential equations \cite{hairer_theory_2014,gubinelli_paracontrolled_2017}. Due to these technical breakthroughs, the stochastic quantization program for constructing Euclidean field theories has been carried out for the $\Phi_3^4$ theory, and part of the program has been carried out for pure Yang-Mills and Yang-Mills-Higgs in 2d and 3d \cite{chandra_chevyrev_hairer_shen_ym2d_2022}.

One approach to the problem is by taking continuum limits of lattice Yang--Mills model. In four space-time dimensions, continuum limit convergence for the Abelian lattice theory was proven by Driver \cite{Driver87}, and significant progress towards convergence for the non-Abelian theory was made by Balaban (i.e. \cite{balaban_large_1989,balaban1987renormalization}). The two dimensional non-Abelian pure Yang--Mills theory is relativey easy to construct via gauge symmetry, however constructing even the 2d Yang--Mills--Higgs theory or 3d pure Yang--Mills theory in the continuum is still open

Moreover, since Wilson's proposal to define non-Abelian gauge theory nonperturbatively via Euclidean lattice regularization \cite{wilson1974renormalization}, lattice gauge theory has become a very important tool for strongly-coupled Yang-Mills. For instance, a simple cluster expansion argument shows that at very strong coupling, lattice Yang--Mills theories exhibit a mass gap and area law \cite{osterwalderseiler1978}, and more recently gauge--string duality and stochastic quantization based approaches weakened the requirement on the coupling \cite{shen2023stochastic,cao_expanded_2025,CNS2025b,Nissim25}. In parallel, the large-$N$ limit has has become a complementary organizing principle for nonperturbative Yang-Mills, where lattice computations produce quantitative evidence that many infrared observables admin a smooth extrapolation as $N= \infty$. The pure Yang--Mills $N=\infty$ limit has also recently been proven to exist and has been described rather explicitly in terms of rigorously defined surface sums or string trajectories \cite{Chatterjee2019a,BCSK2024}. These conceptual and computational advances underpin increasingly precise nonperturbative determinations of Yang–Mills observables (e.g. glueball spectroscopy and related continuum extrapolations) \cite{MorningstarPeardon1999}. 

From a constructive perspective, lattice gauge theory is compelling because it supplies an ultraviolet-regularized, gauge-invariant definition of the theory at finite lattice spacing; the central mathematical challenge is to understand (and ultimately prove) the continuum limit and its universality class. In four dimensions this is intimately tied to the Clay Millennium Yang–Mills existence and mass gap problem \cite{jaffe_quantum_2000}. What is missing is a fully rigorous control of the renormalization/continuum-limit map for non-Abelian gauge theories that would (i) construct a continuum Wightman/OS-compatible theory from the lattice, (ii) identify the emergent long-distance physics (confinement, mass gap, string tension), and (iii) prove that physically relevant observables are independent of microscopic discretization choices (universality). Even at a more structural level, lattice gauge theory exposes deep issues mathematicians can engage with: reflection positivity and reconstruction (for connecting Euclidean and Minkowski formulations), the analytic structure of RG flows, and the subtle interplay between topology and locality (e.g. topological charge sectors becoming sharply separated in the continuum as illuminated by Wilson flow) \cite{Luscher2010}. Finally, the fermionic sector poses “mathematics-first” problems: implementing chiral symmetry and anomalies without doublers \cite{GinspargWilson1982,Luscher1998,Neuberger1998} is now conceptually clean but still raises hard questions about locality, index theorems at finite cutoff, and extensions toward genuinely chiral gauge theories. A classic reference framing many of these issues explicitly as problems in constructive QFT/statistical mechanics is Seiler’s monograph \cite{Seiler1982}, and the foundational lattice formulation and its statistical-mechanical structure are laid out in Wilson/Kogut \cite{Wilson1974,Kogut1979}.

\subsection{Related work}
\label{ss:related}

While we know of no previous formalizations of constructive QFT, a significant related project is the
Brownian motion project of Remy Degenne {\it et al} \cite{degenne_formalization_2025}, 
{\tt RemyDegenne/brownian-motion} .  This contains the Kolmogorov Extension theorem, in  
{{\tt RemyDegenne/kolmogorov{\textunderscore}extension4}}, which is imported by our proof of Minlos.

The two axioms we initially introduced in our formalization, Minlos theorem and one to prove nuclearity of Schwartz space, were proven in our subsequent formalization. Moreover, the dependence on Goursat's theorem was also removed in our subsequent formalization.
At this writing there are three axiom-free constructions.
Matteo Cipollina
has independently formalized the proof of nuclearity and a construction of the Gaussian measure to obtain an axiom-free 
proof of OS in {\tt or4nge19/OSforGFF}.
The first author (MRD) has independently formalized the proofs of these theorems and
adapted our formalization to use them, in the repos {\tt mrdouglasny/gaussian-field} and {\tt GFF4D}.
In addition the repo {\tt mrdouglasny/bochner} has proofs of the Bochner and Minlos theorems, which
are imported by our final released construction in {\tt mrdouglasny/OSforGFF}.   
The original construction with 3 axioms is preserved in the ``original'' branch of that repo.

\section{Axiomatic QFT and Constructive Frameworks}

To rigorously prove that a QFT exists, one must define the properties that a mathematical object must possess to be considered a QFT,
and then construct such a mathematical object.
In this section, we will give an overview of some of the most influential approaches to developing an axiomatic definition on a QFT. 

In physics, a QFT can be defined in several ways: a quantum mechanical system with Lorentz covariance and causality; a path integral over fields
on Minkowski space-time; the related path on Euclidean space-time obtained by Wick rotation; and other possibilities.  Each way can be codified
by an axiomatic approach, in which one specifies the data (or set of observables) that determine the QFT, and a short list of axioms the data must
satisfy.  The axioms include restatements of the physical principles and consistency conditions,
and additional ``technical'' assumptions which facilitate the mathematical arguments.
One can then prove reconstruction theorems which, given the QFT data of one axiomatic approach, show that the data in another approach
is uniquely determined and satisfies the axioms of that approach.

In the rest of the section we will survey some of the most influential approaches to axiomatizing QFT and explain their relations.
The three main approaches are the Wightman axioms \cite{garding_wightman_1964}, the Osterwalder-Schrader (OS) axioms \cite{OS1972}, and the Haag-Kastler axioms \cite{haag_kastler_1964}. Both the Wightman and Haag-Kastler axioms formulate a QFT in Minkowski spacetime, the first in terms of correlation functions and the second in terms of $C^*$ operator algebras. The OS axioms axiomatize QFT in Euclidean space-time.  Other important approaches are the CFT approaches using operator product expansions or sewing axioms, and factorization algebras.

The most direct approach is to formulate the QFT in the usual language of quantum mechanics, with a Hilbert space, a Hamiltonian and field operators.  These must satisfy the Wightman axioms, which express Lorentz invariance, microscopic causality, and the existence of a vacuum state. Physically,  this is the original definition of QFT; the other approaches are justified by their relation to it. Assuming the Wightman Axioms, one can rigorously prove the CPT theorem and the spin-statistics theorem \cite{streater_wightman_pct_1964,jost_general_theory_1965}.

Given field operators and a vacuum state, one can define the vacuum expectation value of a product of $n$ operators,
the Wightman distributions on $M^n$.  The Wightman reconstruction theorem shows that this data uniquely determines the original QM formulation \cite{garding_wightman_1964}. Wightman distributions can be analytically continued to complex time and thus to Euclidean spacetime. The analytically continues objects are called Schwinger functions, which are non-singular on n-tuples of points in $M^d$ which are pairwise distinct \cite{OS1972}. Schwinger functions are the defining observables of Euclidean QFT, and can be summarized in a generating function  (\ref{eq:genfn}).

In the Euclidean approach to constructive QFT one typically requires the Schwinger functions, or equivalently a generating functional that determines the Schwinger functions to satisfy the Osterwalder-Schrader axioms \cite{OS1972}. The corrected Osterwalder-Schrader reconstruction theorem \cite{OS1975} shows that the Schwinger functions obeying the Euclidean axioms together with an additional growth condition are sufficient to ensure temperedness and admit an analytic continuation to Minkowski-space Wightman distributions which satisfy the Wightman axioms. Hence, to be precise the 1975 reconstuction theorem revises the original 1973 Osterwalder-Schrader formulation by adding an extra regularity/growth input. In this framework, reflection positivity is the crucial Euclidean remnant of Hilbert space positivity after Wick rotation. A stronger but sufficient framework is the Glimm and Jaffe \cite{glimm2012quantum} measure theoretic formulation, which assumes a Euclidean probability measure on a space of distributions and imposes stronger regularity conditions adapted to constructive proofs. While this is more restrictive than the abstract Schwinger function setup, it is often more technically convenient in application, and certainly more friendly to formalization efforts. 

In Haag-Kastler (or algebraic QFT) one replaces point-like fields by a net of local observable algebras where to each suitable spacetime region one assigns a $C^*$ algebra (for a fixed representation of a von Neumann algebra) of observables localized in spacetime \cite{haag_kastler_1964}. In a vacuum representation one further assumes a vacuum vector that is cyclic for the quasilocal algebra together with the positive conditions for the translation group. 

In conformal field theory (CFT), it is often natural to take as fundamental the spectrum of local operators (e.g.\ primaries) and their operator product expansion (OPE).
Schematically, for local operators $\mathcal{O}_i$ one posits an expansion
valid inside correlation functions, where conformal symmetry fixes the spacetime dependence
of $C_{ij}{}^{k}$ up to numerical OPE coefficients and representation-theoretic data
\cite{BPZ1984, PolandRychkovVichi2019}. In unitary CFTs one can make this precise: the OPE and conformal-block decomposition
converge inside correlators in an appropriate domain, with a finite radius governed by the operator-insertion
geometry \cite{PappadopuloRychkovEspinRattazzi2012}.
Consistency of performing OPEs in different channels for a four-point function yields crossing symmetry constraints (bootstrap equations), which can be treated analytically
in special cases and numerically in general \cite{RattazziRychkovTonniVichi2008,PolandRychkovVichi2019}.

In two dimensions one can also formulate CFT globally via sewing/gluing axioms. Segal’s functorial axioms encode the assignment of state spaces to boundary circles and amplitudes (linear maps) to Riemann surfaces with parametrised boundaries, subject to compatibility
with gluing of surfaces \cite{Segal1988}. Related sewing constraints provide field-theoretic control of how
correlators behave under cutting and gluing, and play a key role in rational CFT and the Moore--Seiberg
consistency conditions \cite{MooreSeiberg1989}.

 The primary reason we formalized Glimm-Jaffe \cite{GJ1987Axioms} rather than other approaches is that their axiomatic framework turns the problem of defining a QFT into that of defining a functional measure (a path integral), which in many interesting cases (scalar fields with a real action, pure Yang-Mills theory) is a probability measure (real and non-negative). 
 Measure theory and probability are relatively well covered by Lean Mathlib. 
 By contrast, Mathlib currently lacks a sufficiently developed theory of operators on Hilbert spaces, 
 as well as sufficient complex analysis and Riemann surface theory for the 2d CFT approaches.

In this paper, we constructed the free scalar field using the GJ construction and proved that it satisfies the given axioms. This is an important proof of concept that given the existing Lean libraries (Mathlib) and AI tools, such construction is already possible, and will only be more efficient in the future. Next milestones would be to prove the OS reconstruction theorem \cite{glimm2012quantum}, to construct free fermions, and interacting theories in two dimensions.

\section{Introduction to the mathematics}
\label{s:math} \label{ss:framework}

We generally follow the notation of Glimm and Jaffe \cite{GJ1987Axioms}, and our precise notations are in appendix \ref{app:notation}. 
We treat a free massive scalar field in $D=4$ Euclidean dimensions;
while the strategy and arguments work in any dimension, some results on the covariance would require adaptation for other dimensions.

In physics terms, we want to define a functional measure which naively looks like
\begin{equation}\label{eq:gff-path-integral}
    d \mu[\phi] \sim \prod_x d\phi(x) \exp (-S[\phi])
\end{equation}
with the action 
\begin{equation} \label{eq:gff-action}
    S[\phi] = \frac{1}{2} \int_M d^Dx\, (\partial\phi)^2 + m^2\phi^2
\end{equation}
with $m^2>0$. This is naive as Eq. \eqref{eq:gff-path-integral} {\it cannot} be interpreted as a product of an infinite dimensional Lebesgue measure 
$\prod_x d\phi(x)$ with the weight $\exp( -S[\phi])$ as the first of these does not exist.  Rather, the measure
must be constructed by other means.

The Glimm-Jaffe axioms are stated in terms of the generating functional of correlation functions,
\begin{equation} \label{eq:genfn}
    S(J):= \int e^{i \int_M d^Dx\, J(x) \phi(x) } d \mu[\phi] .
\end{equation}
In a physics treatment we could immediately write it down for the free field,
\begin{equation} \label{eq:freegenfn}
    S(J):= \int e^{ic \int_{M\times M} d^Dx\, d^Dy\, J(x) K(x-y) J(y) } 
\end{equation}
where $K(x-y)$ is the propagator (Green function for $\Delta+m^2$) or covariance (in the language of probability theory).

However, from a rigorous and even more so from a formal perspective, we have things to define and prove. What is $K$?
Under what conditions on $J$ does Eq. (\ref{eq:freegenfn}) make sense?
But the first order of business is to show that there exists a measure $d \mu[\phi]$
for which Eq. (\ref{eq:genfn}) equals Eq. (\ref{eq:freegenfn}).
Conceptually this is not difficult; essentially, we want to diagonalize $K(x-y)$
and define $d\mu$ as the product of one-dimensional Gaussians over its
eigenspaces.  However there are some technical obstacles to doing this which we discuss below and which led us (in version 1) 
to appeal to
a general existence theorem due to Minlos for this step.

Suppose we have this measure, next we need to show that this statistical
field theory can be related to a true quantum field
theory with a Hilbert space, Hamiltonian, and Lorentz symmetry.  This follows from the 
Osterwalder-Schrader reconstruction theorem \cite{OS1973} (and subsequent work \cite{OS1975}).
In these works it was shown that this will be the case if Eq. (\ref{eq:freegenfn}) satisfies a short
list of ``OS'' axioms.  There are variations on the original axioms, we follow the
set in Glimm-Jaffe.
\begin{itemize}
    \item OS0: Analyticity- The functional $S (f) $ is analytic. More precisely for every finite set of test functions $f_j \in \mathcal{D}(\R^D)$ with $j= 1,2, \ldots, N, $ and taking $\vec z = ( z_1, \ldots, z_n) \in \C^n$, the function 

\begin{equation}
    (z_1,\dots,z_n) \rightarrow S \left(\sum_{j=1}^n  z_j f_j \right)
\end{equation}

is entire on $\C^N$.

    \item OS1: Regularity-
There exists a real number \(p\) with \(1\le p\le 2\) and a constant \(c>0\) such that, for every test function $f \in \mathcal{D}(\R^n)$,
\begin{equation}\label{eq:OS1}
\bigl|S\{f\}\bigr|
\;\le\;
\exp\!\Bigl[c\bigl(\|f\|_{L^{1}}+\|f\|_{L^{p}}^p\bigr)\Bigr].
\end{equation}

If \(p=2\), we further assume that the two‑point Schwinger function
\[
S_{2}(x,y)\;:=\;
\left.\frac{\partial^{2}}{\partial z_{1}\partial z_{2}}\,
S\!\bigl\{\,z_{1}\delta_{x}+z_{2}\delta_{y}\bigr\}\right|_{z_{1}=z_{2}=0},
\]
is locally integrable in $(x,y)$ and is defined in our generating functional, $S(f) = \int e^{i \phi(f)} d \mu$. Let $d \mu$ be a probability measure on $\mathcal{D}'(R^D)$ which satisfies OS0. Then the measure $d\mu$ has moments of all order. The nth-moment has a density $S_n(x_1, \ldots, x_n) \in \mathcal{D}'(R^{D})$ which means 
\begin{equation}
    \int \phi(f_1) \ldots \phi(f_n) d \mu = \int S_{n}(x_1, \ldots x_n) \prod_{i=1}^n f_i (x_i) dx 
\end{equation}
is locally integrable in \((x,y)\in\mathbb{R}^{D}\times\mathbb{R}^{D}\).Note that for $p=2$ then $||f||_{L^2}$ gives us a Hilbert space norm which is useful for Fourier analysis and two-point functions. 

    \item OS2: Euclidean Invariance- $S(f)$ is invariant under Euclidean symmetries (translations, rotations and reflections). More explicitly, for any such Euclidean isometry of $\R^D$, $E$, we have,
    \begin{equation}
        S(f)= S(Ef).
    \end{equation}

    Equivalently, $d\mu$ is Euclidean invariant  in the sense that $\mu = E_* \mu$ where $E_*\mu$ denotes the pushforward of $\mu$ under $E$.
\end{itemize}

Before moving on to OS3 and OS4, we isolate a class of functions on  $\mathcal{D}'(R^D)$,
      \begin{equation}
        \mathcal{A} = \{ A(\phi)= \sum_{j=1}^N c_j \exp (\phi(f_j)):\, c_j \in \C,\, f_j \in \mathcal{D}'(\R^D))\}.
    \end{equation}

\begin{itemize}
    \item OS3: Reflection Positivity- 
    The Euclidean symmetries act on $\mathcal{D}'(R^D)$, so that $E \phi(f) = \phi(Ef)$
    , which defines a unitary condition action on the $L^2(\mu)$. 
     Finally, let $\mathcal{A}^+ \subset  \mathcal{A} $ be the set of functionals where $f_j \in C_0(\R_{+}^D)$ and $C_0(\R_{+}^D)$ is the space of continuous functions that vanish at infinity in the positive half space. $\R_+^D := \{ (t, \vec{x}): t > 0\}$ so now if we let $\theta: (t , \vec{x}) \rightarrow (-t, \vec{x})$ be the time reflection. Then the content of this axiom states that the inner-product inequality is satisfied:  
     \begin{equation}\label{eq:OS3}
         \int (\theta A)^- A\,d\mu=\langle \theta A, A \rangle_{L^2(\mu)}\geq 0.
     \end{equation}
   
    \item OS4: Ergodicity-  The time translation subgroup $T(t)$ acts ergodically, meaning it upholds the uniqueness of the vacuum, on the measure $d\mu$. The time translation subgroup in Euclidean space is the one parameter group of pure time shifts 
    \begin{equation}
        T(t) : x = (t, \vec{x})\rightarrow (t+ t', \vec{x}) \qquad t'\in \R 
    \end{equation}
so that $T(t+s)= T(t) T(s)$. 

This is equivalent to saying that for all $L^1(\mu)$ functions $A(\phi)$, 

\begin{equation}\label{eq:Glimm-Jaffe-OS4}
    \lim_{t \rightarrow \infty} {1 \over t} \int_{0}^{t}T(s) A(\phi) T(s)^{-1} ds = \int A (\phi) d \mu(\phi),
\end{equation}
where the limit above should be taken in the $L^1$ sense.

\end{itemize}

\begin{remark}[Schwartz functions vs compactly supported test functions]
    In OSforGFF, instead of working with the spaces $\mathcal{D}(\R^D)$ and $\mathcal{D}'(\R^D)$, our axioms and proofs were stated in terms of Schwartz functions $\mathcal{S}(\R^D)$ and tempered distributions $\mathcal{S}'(\R^D)$. Since $\mathcal{D}(\R^D) \subset\mathcal{S}(\R^D)$ and $\mathcal{S}'(\R^D) \subset\mathcal{D}'(\R^D)$, our axioms imply the ones stated in Glimm and Jaffe.
\end{remark}

\begin{remark}[OS4 Ergodicity Statement]
    When Glimm and Jaffe write $T(s) A(\phi) T(s)^{-1}$, they interpret $A(\phi)$ and $T(s) A(\phi) T(s)^{-1}$  as linear operator on the space of functions $\mathcal{D}'(\R^D) \to \C$ where $A(\phi)$ is a multiplication operator. Unpacking the definitions one can check that \eqref{eq:Glimm-Jaffe-OS4} is equivalent to
    \begin{equation}
        \lim_{t \rightarrow \infty} {1 \over t} \int_{0}^{t}A(T(s)\phi) ds = \int A (\phi) d \mu(\phi),
    \end{equation}
    where we now treat $A$ simply as a function on $\mathcal{D}'(\R^D)$.
    
    In {\tt OSforGFF}, the OS4Ergodicity axiom is stated only for functions in $\mathcal{A}$ instead of for all $L^1(\mu)$ functions, and we use $L^2$ convergence instead of $L^1$ convergence. Our axiom implies Glimm and Jaffe's axiom, since a general functional analytic argument shows that $\mathcal{A}$ is dense in $\mathcal{S}'(\R^D)$ with respect to $L^1$ norm, and $L^2$ convergence implies $L^1$ convergence on probability spaces.
\end{remark}

\begin{remark}[OS4 Clustering]
    In OSforGFF, we also include two alternate stronger versions of the OS4 called OS4Clustering and OS4PolynomialClustering which are stated below. As before $S[f]$ denotes the generating functional of the measure $\mu$
    \begin{enumerate}
        \item (OS4Clustering) For any $f,g \in \mathcal{S}(\R^D)$,
        \begin{equation}\label{eq:OS4-Clustering}
            S[f+T(s)g]-S[f]S[g] \to 0
        \end{equation}
        as $s \to \infty$.
        \item (OS4PolynomialClustering) For any $f,g \in \mathcal{S}(\R^D)$, and any $n>0$, there exists a constant $C_{f,g,d,n}$
        \begin{equation}\label{eq:OS4-Poly-Clustering}
            |S[f+T(s)g]-S[f]S[g]| \leq C_{f,g,d,n} (1+|s|)^{-n}.
        \end{equation}
    \end{enumerate}
    Of course OS4PolynomialClustering implies OS4Clustering. We will sketch in Section \ref{subsection:OS4} why these clustering variations imply the original ergodicity axiom.
\end{remark}

\section{Details of the formalization}
\label{s:details}

Here we comment on some choices we made and points where
we encountered technical difficulties. The full proof is given in appendix \ref{app:osproof}. We strongly encourage the reader interested in more detailed explanations
to clone our repo, fire up a coding agent and ask it your questions.

\subsection{Construction of Gaussian Free Field}
\label{ss:construct}

We construct the Gaussian free field (GFF) as a generalized random field by applying Minlos' theorem to a canonical Gaussian characteristic functional on a nuclear space of test functions. An equivalent ``mode expansion'' construction as an infinite sum of independent Gaussian random variables, which is more familiar to physicists, is explained in \ref{ss:defmot}.

Let $E$ be a real nuclear locally convex topological vector space and let $E'$ denote its continuous dual. We write $\phi(f)$ for the canonical pairing $\phi\in E'$, $f\in E$. A function $\chi:E\to\mathbb C$ is called \emph{positive definite} if for every $n\in\mathbb N$, every $f_1,\dots,f_n\in E$, and every $c_1,\dots,c_n\in\mathbb C$,
\[
\sum_{i,j=1}^n c_i \overline{c_j}\,\chi(f_i-f_j)\ \ge\ 0.
\]

\begin{theorem}[Minlos]
Let $E$ be nuclear. If $\chi:E\to\mathbb C$ is continuous, positive definite, and normalized by $\chi(0)=1$, then there exists a \emph{unique} Borel probability measure $\mu$ on $E'$ such that for all $f\in E$,
\[
\chi(f)\ =\ \int_{E'} e^{\,i\phi(f)}\, d\mu(\omega).
\]
Equivalently, every continuous positive-definite functional on a nuclear space is the Fourier transform (characteristic functional) of a unique probability measure on the continuous dual.
\end{theorem}

Fix a domain or manifold $M$ and choose a space of test functions $E$ (in our case $E=\mathcal S(\mathbb R^d)$). Let
\[
C:E\times E\to \mathbb R
\]
be a continuous, symmetric, positive semidefinite bilinear form, interpreted as the \emph{covariance form}. The GFF is defined by
\begin{equation}
    C(f,f)\equiv \int \frac{m}{4\pi^2|x-y|}K_1(m|x-y|)f(x) f(y) \,dx\,dy
\end{equation}
and 
\begin{equation}
    \chi(f)=\exp\left(-\frac{1}{2}C(f,f)\right).
\end{equation}

\subsection{Schwartz Test Functions}

Lean Mathlib already has the definition of Schwartz functions and many of their properties.
Examples of properties we had to prove include double mollifier convergence and various Fubini lemmas. The former allows the definition of Schwinger two-point functions, which is an integral involving a singularity at $x=y$ in $K(x,y)$. The Fubini lemmas are necessary to exchange orders of integration, e.g. spacetime integral and momentum integral, and they require the absolute convergence of integrands. The convergence of regular Schwartz function integrals are easy as we can impose arbitrary polynomial bound; yet, there are subtleties when the free covariance (Green's function) is involved, which is discussed in appendix \ref{app:osproof}. Some of these are general results which could be ``upstreamed'' into Lean Mathlib (work in progress).

\subsection{Motivation for the definitions}\label{ss:defmot}

For a physics reader there is already a lot to unpack and motivate, in particular, the function spaces and other unfamiliar terminology.  Where does this come from?  
We will not try to give a complete answer here, but let us explain the motivation for the use of Schwartz spaces, the concept of Gel'fand triple, and the
Minlos theorem.

First, what is $d\mu(\phi)$?  Mathematical probability theory is based on
measure theory (Kolmogorov); the probability that a random variable $x$ will
take values in a set $S$ is $\int_S d\mu(x)$.  For us the random variable is the
field configuration $\phi(x)$, which (naively) is a function from spacetime $M$ to $\IR$.
We will define the set of all of its possible values as the configuration space $\CC$
(the type {\tt FieldConfiguration} in our formalization).

The way that physicists usually construct an infinite-dimensional field space is to use a series of finite dimensional spaces $\CC_N$, which approximate $\CC$ arbitrarily well as $N\rightarrow\infty$.  A standard choice is to discretize by only keeping the function values at lattice points $\IZ^D\cdot a$
for some lattice spacing $a$, and  take (say) $a=2^{-N}$.  We would also have to take the
total volume finite and explicitly take the infinite volume limit.
We could then define the
measure by taking the product Lebesgue (uniform) measure over all the lattice points,
approximate the action in a standard way say by taking
$\partial\phi \rightarrow (\phi(x+a)-\phi(x))/a$, and then write an
explicit Gaussian weight.  This is perfectly rigorous and a viable approach,
why should we consider alternatives?

One answer is that proving the existence of the limiting measure is nontrivial, and this does not save as much work as one might think.
Another is that, while approximations are not intrinsically problematic, it is preferable that they preserve the very properties one ultimately wishes to prove. In particular, a lattice approximation breaks the full Euclidean invariance appearing in OS2 from the outset. Since there are no nontrivial finite-dimensional representations of the Euclidean group, preserving OS2 naturally leads one to work with measures on infinite-dimensional function spaces. 

A conceptually simple construction, sufficient at least in the Gaussian case, is to diagonalize the covariance (equivalently, the quadratic kinetic form) and write
\begin{equation} \label{eq:infbasis}
    \phi(x) = \sum_{k\in\mathbb{N}} c_k \phi_k(x),
\end{equation}
where the $c_k$ are independent centered Gaussian random variables of variance $1$, and the mode functions $\phi_k$ are chosen so that
\[
K(x,y)=\sum_{k\in\mathbb{N}} \phi_k(x)\phi_k(y)
\]
reproduces the desired covariance kernel.

However, this already raises a nontrivial issue: the covariance is represented by an infinite sum or integral which need not converge pointwise. In the continuum theory this is not an accident but a manifestation of ultraviolet singularity. Indeed, for the free massive field one has in momentum space
\[
K(p)=\frac{1}{p^2+m^2},
\]
and in $D=4$ this is not trace class, since formally
\[
\int_{\mathbb{R}^4}\frac{d^4p}{p^2+m^2}=\infty.
\]
Correspondingly, the position-space covariance $K(x,y)$ is singular on the diagonal, so in particular $K(x,x)$ diverges. Thus the Gaussian field cannot be realized as an ordinary random function, but only as a random distribution. 

One might think that this prevents constructing
such a measure, and this would be the case if we required it to have
support on continuous functions, or even in the Hilbert space of functions on spacetime $L^2(\IR^4)$.
This corresponds to the (probably familiar) fact that a typical value of $\phi(x)$
in dimension $D\ge 2$ is not a continuous function.

The solution to this problem is to broaden the support of the measure to 
distributions, i.e. linear functionals on some restricted space of test functions.
What space?  We need test functions whose values can be freely chosen on any finite set of points;
we need them to have finite integrals against $K(x-y)$, and we need their integral over all of spacetime
to be finite (since $K$ is translation invariant). 

The space $\mathcal{S}(\mathbb{R}^D)$ of Schwartz test functions 
is a common choice to satisfy these requirements, and is well implemented in Mathlib.
These are functions with rapid (faster than polynomial)
decay in both position and momentum space, and the spaces of Schwartz test functions and of distributions are preserved under Fourier transform.
One can describe them concretely using a basis of Hermite functions (normalized harmonic oscillator number states) in an 
expansion analogous to Eq. (\ref{eq:infbasis}),
\begin{equation} \label{eq:hermbasis}
    \phi(x) = \sum_{k\in\mathbb{N}} c_k \phi_k(x),
\end{equation}
The test functions then have $c_n\rightarrow 0$ faster than any polynomial in $n$.
The dual Schwartz distributions
can have $c_n$ growing at most polynomially--to be precise, there exists $k$ such that
\[
\sum_{n\ge 0} |c_n|^2 (1+n)^{-k} < \infty .
\]

While the space of Schwartz distributions $\mathcal{S}'(\mathbb{R}^D)$ is large enough to contain our field configurations,
we still need to show that the sum in Eq. \ref{eq:infbasis} converges with probability one.
This places conditions both on the space (the sum must be over a countable basis) and on the covariance.

A key additional structural fact is that $\mathcal{S}(\mathbb{R}^D)$ is a \emph{nuclear} Fr\'echet space. Informally, nuclearity means that the topology can be described by a nested family of Hilbert norms whose inclusion maps are extremely compact (indeed trace-class in a suitable sense). This is precisely the infinite-dimensional regularity needed to extend finite-dimensional Gaussian intuition to the continuum. In particular, by the Bochner--Minlos theorem, every continuous positive-definite functional on a nuclear space is the Fourier transform of a unique probability measure on its topological dual. Hence a Gaussian characteristic functional of the form
\[
\chi(f)=\exp\!\left(-\frac12\,C(f,f)\right), \qquad f\in \mathcal{S}(\mathbb{R}^D),
\]
defines a probability measure on $\mathcal{S}'(\mathbb{R}^D)$. In this sense, nuclearity plays the role that tightness or compactness plays in finite dimensions: it prevents probability mass from ``escaping to infinity'' when one passes from finite-dimensional approximations to the full infinite-dimensional limit, and thereby underlies the existence of the Euclidean field measure \cite{glimm2012quantum}.

The condition on the covariance can be stated in several ways: one is that it must be continuous in the nuclear topology.   More concretely,
it must be definable by means of a linear operator $T$ which densely embeds the test function space into a Hilbert space,
\be
C(f,g) = \langle T\;f, T\;g \rangle,
\ee
and such that $T$ is a Hilbert-Schmidt operator, $\tr T^\dag T\le\infty$.  This would be false for
$T=(k^2+m^2)^{-1/2}$ considered as a  bounded operator on $L^2(\IR^D)$, 
but if we instead consider it as an embedding of Schwartz test functions
into $L^2$ it is true.   This structure of a test function space with high regularity embedded in a Hilbert space
which itself is embedded into the space of distributions is called a Gel'fand triple (or rigged Hilbert space)
and is foundational for constructive QFT.

For the future (renormalization of interacting theories) 
we want to use smaller spaces of distributions with more regularity.  Despite the seeming abstractness one could
be more concrete about these spaces, for example work on rigorous stochastic quantization uses
Sobolev and weighted Sobolev spaces.  One can still construct the measure
on $S'$ and later argue that the measure has support on a smaller space.  However these function spaces are not yet in Mathlib.

In the original formalization described here, we dealt with these matters by postulating the Minlos theorem and the nuclear property of Schwartz space
as axioms, and carefully checked them by hand. 
As discussed in \S \ref{ss:related}, we now also have axiom-free
proofs which are available in the latest version of our repository.

\subsection{Other infrastructure}

\begin{itemize}
    \item We defined the Euclidean group action in the file \texttt{Spacetime/Euclidean} without much help from Mathlib. We defined the group structure, elements, and their actions on the measure or test functions. Yet these definitions are specified to 4D spacetime and scalar test functions, therefore future work would include rewriting it to standards befitting a foundational element, and generalization to non-scalar fields.

    \item The standard definition of the free covariance is the Fourier transform of the free propagator $1/(k^2+m^2)$, but this is a conditionally convergent integral whose integrability is hard to prove. Instead, we used the Bessel function definition as stated in \ref{eq:cov}.

    \item We used two slightly different definitions for Fourier transform: the physics convention in \texttt{Parseval.lean} and the Mathlib convention in \texttt{General/FourierTransform}, differ by a factor of $2\pi$ in the phase. Although one could stick to one convention, we found it easier and more intuitive to use different conventions when proving different theorems. Unlike symbolic math, the prover will catch the mistake if the wrong one is used, so the mixing of conventions does not cause trouble.

    \item We prove both the positivity required by Minlos and reflection positivity using the Schur-Hadamard theorem, which reduces positivity of the generating functional to positivity of the free covariance. The positivity here is defined in the sense of probability theory, details can be found in appendix \ref{app:osproof}.

\end{itemize}

\section{Methods}
\label{s:methods}

Two of us (MRD and SH) started the project in July 2025 by writing an informal outline of the
math and physics concepts, the OS axioms, and proof strategies.  This was in no sense a ``blueprint''
as it was more sketchy than the published proofs!  Nevertheless we gave it to two models (GPT-5 and
Gemini~2.5) and asked them to translate it to Lean.  The results were not of high quality but they
were helpful in suggesting definitions and theorems which we would eventually need to make.  Still,
the results were better than we had expected given the paucity of Lean code in the training corpus.

We then hand-coded the basic definitions, such as \lstinline{FieldConfiguration} (Schwartz
distributions on space-time), the Euclidean group, and the generating function; everything which was
needed to formulate the OS axioms.  Which we then proceeded to do, again by hand-coding.

From that point on we primarily used coding assistants, at first mostly GPT and Gemini in GitHub
Copilot, with Lean-LSP MCP, which we found to be a significant help.  We also tried Claude Sonnet and
Opus, and while over the course of the fall all of the models improved, in our estimation the
introduction of Claude Code (which launched in VS~Code at the end of September) was the single
biggest jump.  By the end of the project, Claude Code Opus~4.6 (often in the CLI version) was our
go-to coding agent, and it was substantially stronger than the agents we had started the project
with.

We also developed better prompts; one of note deals with the observation that all of the coding
agents seem too willing to abandon a difficult proof and delete their partial work.  We speculate
that while abandoning an unpromising path is often appropriate, their stopping criteria (perhaps
length of time spent?) were tuned to optimize more common coding tasks rather than theorem proving.
Anyways, if one could catch the agent in time, it was useful to say ``Don't give up and delete a
sizable block of code unless you can identify a blocking difficulty, and if you do, explain the
difficulty and lessons learned in the chat.''  This behavior became less of a problem over time,
perhaps due to better agents or better system prompts.

Some parts of the project (such as the Euclidean group action) were forward coded, and others were
elaborated by defining intermediate lemmas by hand and asking the models for proofs.  In attempting
to fill in the proof from there, we developed the backward-chaining strategy by which much of the
project was done: we asked the agents to prove a lemma or OS axiom, but when they ran into
difficulty, rather than leave \lstinline{sorry}s we suggested that the agent propose a simpler
``helper lemma'' which if assumed as an axiom would enable a proof.  Often we gave hints about what
this might be (we were not trying to get an automatic system), but just as often the agent would make
a sensible proposal on its own and formulate it in Lean, and we would accept it.  Except for easy
cases, we added these helper lemmas as Lean axioms rather than \lstinline{sorry}s, so that the model
would stay focused on the original problem, and also added them to a list of axioms to prove later.
As the project filled in this strategy became more and more effective.

Of course, the danger of this procedure is that the conjectured helper lemmas might be false or too
difficult to prove.  Thus it was very important to validate them before proceeding.  In some cases
this could be done by spot checking numerical examples, and we both asked the models to do this and
observed them do it spontaneously.  But our main tool for doing this was cross-validation by the
other models.  A conjecture by Claude would be given to Gemini and GPT (we developed MCP tools to
make this easy) and their comments relayed back to Claude.  Difficult proofs would also employ this
collaborative approach, with the model writing a proof plan which would then be reviewed by the other
models.  Although the relative abilities of the models varied with the versions, we generally found
that Claude was the most effective at coding, with Gemini giving crucial advice as it had a better
and broader understanding of the underlying math and physics.

A variation of the method which we found particularly effective in the later stages was to tell the
agent to look for a helper lemma that is as general as possible, which might appear in a textbook and
which is formulated without using any project-specific definitions.  For example, the statement
should not refer to \lstinline{SpaceTime} but instead to \lstinline{EuclideanSpace R d} (where \lstinline{R} denotes $\mathbb{R}$) and work for
any~\lstinline{d}.  This both significantly decreased the number of false conjectures and also led to
conjectures which were easier to prove.

Several times during our work we discovered that definitions (including our original renditions of the OS
axioms themselves) were incorrect.  A noteworthy example was that we originally defined the
covariance as a Fourier transform of a simple momentum space expression---this is the universal
practice in physics and we did not give it a second thought.  But when a mathematician (RN) joined
the project, he soon noticed that this was a conditionally convergent integral and very problematic
(it could even be interpreted as zero in certain contexts!).  This was soon fixed by introducing an
absolutely convergent definition (in terms of a Bessel function), but the experience led us to
introduce further systematic checks, for example that integrals are absolutely convergent.  The
coding agents made such checks straightforward and were good at implementing special-purpose code
(usually in Python) to speed them up.

The choice of the final axioms to leave as assumptions was made based on prior math knowledge and a
sense of what topics are lacking in Mathlib, as well as repeated proof attempts.  We envisioned
assuming Minlos' theorem from the very start for the reasons discussed earlier.  The other two
``axioms'' (Goursat's theorem and the nuclear property of Schwartz space) were attempted, and the
length of the proof plans and estimated Lean code dissuaded us. 
As mentioned above, these axioms were subsequently eliminated, the first by explicit proof
and the second by redoing the proof of OS0.

It seems to us that for many
projects, assuming major theorems as axioms is a better course than restricting to only allow
statements which can be entirely proven from the foundations.  Of course this has pitfalls which we will discuss below in \ref{pitfalls}.

\subsection{Tips and common pitfalls} \label{pitfalls}

We emphasize that all of the observations in this subsection reflect the state of AI tooling as of
early 2026. Given the pace of improvement, some of these limitations may well be short-lived.
Indeed, perhaps the most striking aspect of this project is how noticeably the models improved over
its course.  When we started, it was a struggle to get anything nontrivial done.  By around December,
proving things was no longer the main difficulty, but the problems described below, false analogies,
context limits, stale documentation, were serious obstacles.  By the time of writing, we are
getting some fairly nontrivial proofs done with minimal guidance, though other proofs remain hard.

\paragraph{Correctness of definitions.}
Large language models have a reasonable intuition for what is ``correct'' or ``conventional'' in
mathematics, and in our experience, if one asks two models independently whether a conjectured axiom
is true, at least one will usually flag an outright false statement---which is why the cross-model
validation protocol described earlier is so valuable.  The more insidious failure mode is subtler: the
models can be led astray by analogies that are nearly but not quite right.  For instance, a model may
confidently apply results about complex matrices to matrices over a Grassmann algebra, or treat a
conditionally convergent integral as though it were absolutely convergent, because these are the
``standard'' settings it has seen most often in training.  Even when given explicit instructions about
how a proof should be structured, the model's prior intuitions can override the guidance, and it will
repeatedly attempt strategies that do not apply in the non-standard setting.  During the Lean
code-writing phase, the most important role for a human in the loop was therefore to ensure that key
definitions were correct and that the mathematical context was not being silently shifted to a more
familiar one.  Subtle slips of this kind are hard to detect by inspection, especially when the types
still line up and everything compiles.  A reliable countermeasure was to constantly build small unit
tests for newly introduced definitions and lemmas.  One can of course feed Lean code back into an LLM
and ask it to translate the formal statement into natural language, but we found that this step alone
did not reliably catch transcription or interpretation errors.  A more robust debugging strategy was
to take a lemma just proved, or a definition just stated, and use it to prove an additional statement
as a unit test, not necessarily something dictated by the original scaffold.  In practice, one such
auxiliary check was usually enough to surface hidden mistakes and force a correction in the relevant
definitions.

\paragraph{Decomposition of long proofs.}

A shortcoming of our formalization is due to the models' limitations at long-form proofs. This is largely a consequence of context limits: the model cannot
simultaneously ``see'' all definitions, previous lemmas, and local proof state for a large development, nor can it reliably output a long proof script without losing earlier constraints. Our formalization's design, logical flow, and readability still requires much improvement to meet the foundational standards of Mathlib or PhysLean. As the project went on, the refactor for quite large sections of the code became much more manageable. One example is the proof of Bochner's theorem, which was done by Claude Code in February 2026 over the space of two days nearly autonomously. However it still is the case that when pushing the models to plan and execute very long horizon proofs the resulting Lean code often falls into the ``merely true'' rather than highly legible, logical, and production ready unless a significant effort is imposed by a human to require such standards.
Still we know of no clear obstacle to further automating this, using better models and giving them better explanations of the standards.

Making new definitions and using them in long proofs remains a challenge.  Human input is
valuable for formulating definitions, planning the high-level decomposition and supplying a proof outline in ordinary
mathematical language; the AI is most effective when asked to fill in individual steps rather than
invent an entire long proof from scratch.  As a result, at least with current tooling,
formalization tends to follow known human proofs, and it is difficult to use these systems to
reliably produce genuinely new arguments that were not already understood at the informal level.

\paragraph{Stale docstrings and unreliable summaries.}
A related issue is that the LLMs understand the codebase largely through docstrings and comments
rather than by genuinely reading the Lean code.  These docstrings are not always up to date: they may
reference definitions or lemmas that have since been deleted, or record earlier proof strategies in
phrases like ``not \ldots\ but \ldots'' that no longer reflect the current approach.  As a
consequence, when we asked an agent to summarize the contents of a file, the resulting markdown
documents were often polluted with stale or irrelevant information inherited from outdated comments.
Worse, the models showed little ability to distinguish the essential components of a proof from
routine bookkeeping: summaries would dwell on trivial details (type coercions, import boilerplate)
while glossing over or entirely omitting the key mathematical ideas.  This made it difficult to rely
on AI-generated documentation for planning or for onboarding new collaborators, and we found that
human-written or at least human-edited summaries remain necessary for any high-level overview of
the proof architecture.  Again, there is no evident reason that these reflect fundamental shortcomings
of the models, but they will need to be dealt with to further automate the procedure.

\paragraph{Unproductive search and the need for human redirection.}
There were theorems for which the model worked for hours, repeatedly adding auxiliary lemmas and
expanding the codebase without making progress.  The model would propose alternative approaches, but
they often failed to resolve the core obstruction. Again, this is a result of the model's inability to identify the core component or central argument of the proof. When a human stepped in with a clearer proof
route, the LLM could usually implement the new direction quickly, and the development then proceeded
smoothly.

\paragraph{Workflow and complexity management}
As the project grew, several difficulties particular to the later stages emerged: completing
successful chains of arguments and not getting lost in the thicket of helper lemmas, keeping track of
which lemmas and axioms are contributing to the main line of the proof, and occasionally keeping
track of two or three proof strategies being simultaneously developed since none had closed.  This
was aided by the tool \texttt{DependencyExtractor}~\cite{cihal2023leangraph}.  But the main
organizational technique we followed was to frequently ask the agents to write and update auxiliary
documents describing the plans being followed to close each active chain of arguments and the current
situation.  For example, we had a file \texttt{axioms.md} which listed the \lstinline{sorry}s and
assumed axioms.  Such organizational methods will no doubt be further automated in the future. Planning more broadly was essential and deserved even more investment than we gave it.

\section{Going forward}
\label{s:going-forward}

Let us outline several directions that we expect to develop in the near future: the formalization of results in mathematical physics, the prospects for formalizing physics
which presently does not exist in mathematically rigorous form, and the further development of
formal mathematical libraries. We also comment on the role of formalization in a hypothetical future in which
AI's have human or superhuman reasoning ability.

\paragraph{Existing results in mathematical physics}
We start by listing a few existing results for which formalization could be in range.
One is the Osterwalder-Schrader reconstruction
theorem, which relates the Euclidean QFT approach we followed to the original Wightman definition of QFT
as quantum mechanics.  Ideally this would use a version of the OS axioms which are both provable and state necessary 
and sufficient conditions for reconstruction, but this is not completely clear for the current axioms
(a recent discussion is \cite{kravchuk_distributions_2021}).  Formalizing the reconstruction theorem could help resolve
these questions, and is being studied by Xi Yin and the first author.

One of the simplest interacting QFTs is $P(\phi)_2$, the 2d boson with a polynomial potential.  This was
constructed in \cite{glimm2012quantum, Guerraetal1975, Simon1974, Nelson1973} and these proofs look
formalizable. In fact, the first author has made progress on it.  More recently the construction was redone using stochastic quantization \cite{Duchetal2024}.
This should be formalizable given the prerequisite stochastic PDE theory.

Rigorous construction of fermionic field theories is generally easier as they do not have the large field problem,
and there are many works on this.  A rigorous treatment of the 2d supersymmetric Wess-Zumino model is in
\cite{jaffe_two-dimensionaln2_1988}.  The 2d abelian Higgs model was treated in \cite{BrydgesFrohlichSeiler1979,BrydgesFrohlichSeiler1980,BrydgesFrohlichSeiler1981}.  
If these results could be combined to construct the supersymmetric gauged linear sigma model,
one could aspire to a direct QFT argument for mirror symmetry.

Another field that is increasingly being formalized is quantum information theory \cite{Meiburg2025}. Recent work has begun to translate central theorems in quantum information into Lean proofs. For example \cite{Meiburg2025}
formally verified the Generalized Quantum Stein's Lemma, a fundamental result in quantum hypothesis testing that gives an operational interpretation of quantum relative entropy in the asymptotic discrimination of quantum states. Their formalization included components from operator algebras, functional analysis and topology. In the process, the formalization clarified several implicit assumptions in the original analytic arguments and led to a more precise formulation of quantum resource theories, while simultaneously contributing to the development of a growing QuantumInfo library \cite{LeanQuantumInfo} in Lean intended to support a broader program of formalizing quantum theory itself.

From the perspective of formal theorists, it would be very appealing to use these ideas to formulate conjectures about quantum gravity as well. One ambitious possible direction would be to clarify the mathematical structure of the gravitational path integral. The gravitational partition function is written as a sum over spacetime geometries, schematically, $Z = \int \mathcal{D}g e^{i S[g]}$, where the integral ranges over metrics modulo diffeomorphisms. However, in gravity the functional measure $\mathcal{D}g$ is not rigorously defined. Moreover, the configuration space itself, remains poorly understood mathematically \cite{Engle2009}.

In particular, the gravitational path integral involves integration over an $\infty-$dimensional space of metrics (and possibly topologies), and the corresponding measure cannot be treated as a straightforward analogue of Lebesgue measure. Determining the correct measure is highly non-trivial in constrained gauge systems such as in gravity \cite{Engle2009}. Consequently, despite a central role in semiclassical and holographic calculations commonly studied in quantum gravity settings, the gravitational path integral is not a mathematically controlled object\cite{Fraaije2022}. Developing a rigorous definition of this measure (or even clarifying what the appropriate configuration space should be) would represent a significant conceptual advance in quantum gravity research.

\paragraph{Non-rigorous theoretical physics}

The low-hanging fruit in the formalization of physics is to improve the verification of calculations and computations.
Many algebraic calculations in physics are essentially rigorous; they need only be placed in the
right physical and mathematical context.  Formalizing this context, even without
proving it all from mathematical foundations, would be valuable if it enables us to find
errors due to incompatible conventions or other misapplications of results.  A famous example
which might have been caught this way was a long-standing disagreement between theory and
experiment of the muon $g-2$, which turned out to be due to a wrong sign convention \cite{Knecht_2002,Hayakawa:1997rq}.

Going the other way, there is little to be gained by formalizing symbolic algebra calculations
in Lean (or the other ITP languages) as a good symbolic algebra system should be 100\%  reliable.
Rather, one could extend the scope of formalization by implementing communication between ITP verifiers
and symbolic algebra systems, and the importation of certificates validating such calculations.
Recent work on this includes
\cite{rowney2026dsleanframeworktypecorrectinteroperability}.

While valuable, these methods do not help much with formalizing theories whose formulations are far beyond
present-day mathematical physics, such as 4d Yang-Mills theory and its supersymmetric extensions.
Given their central role in theoretical particle physics, excluding them is hardly an acceptable limitation.
Rather, ITP will make it far easier to postulate axioms which capture our physicists' current understanding
and intuitions, try them out and see what consequences they imply.

\paragraph{Formalization of more existing mathematics}
The most immediate enabling factor for these goals is the formalization of more existing mathematics. 
Mathlib, while impressive, is still small relative to the body of known mathematics.  In part this is
because it (rightly) maintains
a very high standard of code quality.  As of this writing there are roughly 2,300 pull requests
waiting to be reviewed, each requiring substantial human effort.  
We hope that AI tools will soon
help alleviate this in at least two ways. First by assisting the review process itself, checking style
conventions, verifying clean integration with existing modules, and flagging potential issues before a
human reviewer looks at the PR. Second, by helping authors bring their individual projects into
compliance with library standards before submission, lowering the barrier to contributing and reducing
the burden on maintainers.

A closely related issue is the fragmentation of independent formalization projects.  As more groups
undertake projects like ours, each might start by defining its own versions of the basic objects it
needs, Euclidean transformations, quantum field theory axioms, function spaces, and so on.  This
could lead to many independent, mutually incompatible formalizations of the same
mathematical structures, making it difficult for one project to build on another's work.  While
solving these problems motivates the Mathlib project, it is evidently hard for any single project to
do this across the full breadth of the mathematical sciences.  It is not clear to us that a single
monolithic library (or even one for math and one for physics) is the best or even a scalable solution.
We can compare the present situation with a 
hypothetical world in which (say) the Annals of Mathematics were the only refereed mathematics journal.
The reader can judge the desirability and practicality of this for him or herself.
Perhaps several libraries building on Mathlib in various more specialized directions would be better.
An alternate approach would be to design lightweight but rigorous standards, shared definitions, interface conventions, and
compatibility requirements, along with search tools to find relevant prior work and adopt its conventions from the start.
This would enable independently developed projects to interoperate without
requiring every one to go through a full human review process.
Designing such standards and
building tools that help authors conform to them, strikes us as a very important
infrastructure problem for the formalization community going forward.

\paragraph{Benchmarking formalization progress}
It would be valuable to develop benchmarks that measure formalization progress itself.
Existing benchmarks such as miniF2F~\cite{minif2f} and ProofNet~\cite{proofnet} evaluate how well AI
systems can prove formalized statements, but they measure the capability of the AI, not the state of
human mathematical knowledge that has been made rigorous. What is missing is a way to track how much
of mathematics has actually been formalized, where the remaining gaps
are, and how they are shrinking over time. Some recent work addressing this problem is \cite{constructionverification}.

Here we note an interesting distinction between physics and math.  In mathematics, formalization
proceeds bottom-up: one starts from definitions, builds lemmas, and works toward theorems, always on
solid ground.  Progress can be measured, at least roughly, by how much of the known corpus has been
formalized from the foundations.  In physics the natural direction is often the reverse.  One starts
from a conjecture or a physical intuition---say, that a certain quantum field theory exists and has
specific properties---and works backward to identify the precise mathematical criteria under which
the conjecture holds.  This is essentially the backward-chaining strategy we employed in this
project, elevated to the level of an entire research program.  A meaningful benchmark for physical
formalization would therefore need to capture not just how many theorems have been proved, but also
how far the top-down analysis has reached: how many of the assumptions have been made explicit, how
many have been validated or closed, and how large the remaining gaps are.  Developing such metrics
would help the community track progress, compare approaches, and identify where effort is most
needed.

\paragraph{Further development of AI}

With the rapid improvement of coding agents, it is becoming
possible to formalize entire chapters of textbook mathematics at a
pace that would have been impractical a year ago. 
Progress on the shortcomings of AI-generated code quality and organization mentioned earlier would be quite valuable from
this point of view.
But despite these shortcomings,
if coding agents continue to improve at the current rate, then autoformalization, the complete automation 
of formalization, could well become practical in the near future.
We can hope that this will see rapid application and solve the library problem fairly quickly.

Still, one can imagine that if AI can so easily translate between informal and formal reasoning,
perhaps it will also do 100\% reliable informal reasoning.  
Recent progress can be read both ways: in the 2025 IMO (considered as a challenge for AIs), both types
of reasoning (informal and formal) led to strong performance.
If informal reasoning becomes reliable, will there be any need for formalization?  

One can of course compare this question with the parallel one about mathematical rigor.
And with few exceptions (theoretical physics being one), the historical evolution in the exact sciences
leads towards more rigor: this is quite clear in applied math, statistics and computer science.
All other things being equal, except in the initial exploratory phase of research, the advantages of rigor seem to win out.
This suggests that if formalization becomes easy, it too would win out.

Another analogy, drawn from computer science, is to the differences between interpreted
and compiled languages.  Here the fundamental difference is whether the computer translates the higher level
code to machine code in advance (using a compiler) or as it is being executed (in an interpreter).  This difference
is also reflected in the language -- to be compilable a language must be more rigid, making coding more difficult,
but in return one gains faster and more secure execution as many problems can be detected in advance.  In
practice one usually uses interpreted languages for smaller short-lived projects, and compiled languages for 
larger long-lived software.  The analogy between compilation and formalization is that both methods do more required
work (of translation; of verification) up front, saving the need to do it again later. 
Thus the advantages of formalization will be more significant in working with
large theories and frameworks which undergo sustained mathematical development.   This clearly includes quantum
field theory, suggesting that even if reliable AI informal reasoning about it becomes possible,
formalization will still bring significant advantages.  

\subsection{Advice for now (early 2026)}
\label{ss:advice}

This reflects our impressions and practice as of March 2026 and may go out of date quickly.
\begin{itemize}
    \item To optimize the coding agents' abilities, the most important resource is context -- reading many large files will rapidly exhaust this.
    Keep the project well organized with many supplementary documentation files:
    summaries, plans, lists of sorrys or custom axioms remaining to prove, and so on.  We find that markdown (md) format is better than latex or pdf for these files.
    \item The frontier models' abilities differ and their ranking changes frequently, so one should use several models and follow their performance.  A pattern which seems more
    stable is that the coding agents with harness (Claude Code and Codex), while better at coding, will also try hard to prove incomplete or even false statements.
    The chat/query agents (especially Gemini) are more able to evaluate definitions and theorem statements for correctness and propose refinements and corrections.
    \item Specific wording of prompts is less important than it used to be.  If the models have the information they need, they can work with high level instructions.
    Skills such as lean4-skills (at \\ {\tt https://github.com/cameronfreer/lean4-skills}) and tools such as Lean-LSP MCP are valuable.
    \item The most difficult phase of a project is the beginning, developing project-specific definitions and a proof sketch.  One approach, and arguably the best approach, if a lot of human time and
    expertise are available, is to do all this by hand.  On the other hand this doesn't scale and we are having reasonable success with letting the agents do this,
    with a good deal of human review and guidance.
    \item It is very advantageous if the precise final theorems to prove can be formulated from the start, with a minimal use of project-specific and non-tested definitions. 
    If too many new definitions are required, the project may be out of range.
    \item One should give the models the actual math/physics papers with the informal proofs (latex is better but pdf works) and ask to study them and make summaries.
    \item Although one can let the models improvise a proof sketch, once it is reasonably filled in (having a clear proof path but with many sorrys or axioms) one should have the
    models make clear summaries and vet these thoroughly before proceeding.
    \item We found a backward chaining approach to work well: the first draft of a proof could use sorrys; ask the agents
    to fill them directly; if any are too hard for this then ask the agent to formulate a ``textbook axiom''
    which would fill them, then carefully vet these proposals with other models and by human review before signing off
    and leaving them to prove later.
    \item Keeping the supplementary docs and lean docstrings up to date pays off.
    \item One needs to use the best (``pro'') models and a professional subscription to do research level formalization.
\end{itemize}

\section*{Acknowledgments}

We thank Jeremy Avigad, Sourav Chatterjee, Martin Hairer, Jesse Han, Arthur Jaffe, Kim Morrison, Dwarkesh Patel, Tripp Roberts, Scott Sheffield, Julian Sonner, Christian Szegedy, Yoh Tanimoto, Joseph Tooby-Smith, Zixiao Wang and Xi Yin for valuable discussions. 
Some early stages of the project were done during the July 2025 visit of MRD and SH to Morph Labs/Math Inc,
whom we thank for inspiration and hospitality.
We thank the many contributors to and maintainers of Mathlib, we thank R\'emy Degenne and Peter Pfaffelhuber for 
formalizing the Kolmogorov extension theorem, and we thank
Matteo Cipollina for his independent formalization of the axioms in our first release. 
Finally, we thank Alok Singh for suggesting an MCP setup and more generally for cluing the authors into the techniques being used by Lean engineers, and David Renshaw (who is supported by ICARM) for code golfing.

\bibliography{airefs,trans,mathrefs,mathqft,qft,lean}

\pagebreak

\appendix

\section{Notations} \label{app:notation}
Here we provide the notations we used in this paper.
\begin{enumerate}
    \item Function Spaces
    \begin{enumerate}
        \item For $\mathbb{F} \in \{\R,\C\}$, $C^{k}(\R^D;\mathbb{F})$ will denote the space of $\mathbb{F}$-valued functions on $\R^D$ for which all derivatives of order $\leq k$ exist and are continuous, and we define the set of smooth functions as $C^{\infty}(\R^D;\mathbb{F}) = \bigcup_{k \geq 0} C^k(\R^n;\mathbb{F})$. For $\mathbb{F}=\C$ we will simply denote the spaces $C^{k}(\R^D;\mathbb{F})$ ($C^{\infty}(\R^D;\mathbb{F})$) by $C^{\infty}(\R^D)$ ($C^{k}(\R^D)$).
        
        \item $C_c^{\infty}(\R^D)$ will denote the set of smooth functions with compact support. We frequently use the notation $\mathcal{D}(\R^D)=C_c^{\infty}(\R^D)$ for the set of test functions. The topology on $\mathcal{D}(\R^D)$ is specified by saying a sequence of functions $f_n \to f$ in $\mathcal{D}(\R^n)$ if $\partial^{\alpha} f_n \to \partial f$ uniformly on $K$, for any multi-index $\alpha$ and compact set $K \subset \R^D$.

        \item The space of Schwartz functions $\mathcal{S}(\R^D;\mathbb{F})$ denotes the set of  $f \in C^{\infty}(\R^D;\mathbb{F})$ for which all derivatives $\partial^{\alpha}f$ decay faster than any polynomial, i.e. $\sup_{x \in \R^n}|(1+|x|)^{-k}\partial^{\alpha}f(x)|<\infty$ for all $k \in \mathbb{N}$ and $\alpha \in \mathbb{N}^n$. The seminorms $\|f\|_{k,\alpha}:=\sup_{x \in \R^n}|(1+|x|)^{-k}\partial^{\alpha}f(x)|$ turn $\mathcal{S}(\R^D;\mathbb{F})$ into a Fréchet space.
        
        \item  For $p\in [1,\infty)$, and a measure space $(X,\mathcal{F},\mu)$ and $\mathbb{F} \in \{\R,\C\}$, \\$L^p(X;\mathbb{F})=L^p(\mu;\mathbb{F})=\{ \text{measurable } f \text{ on } \R^n \text{ taking values in } \mathbb{F}: \int_X |f|^p dx <\infty \}$, and\\
        $L^{\infty}(X;\mathbb{F})=L^{\infty}(\mu;\mathbb{F})=\{ \text{measurable } f \text{ on } \R^n \text{ taking values in } \mathbb{F}: \mathrm{ess\,sup}|f| <\infty \}$
        Once again when $\mathbb{F}=\C$ we omit $\mathbb{F}$ from the notation. $L^p(X;\mathbb{F})$ is a Banach space with norm $\|f\|_{L^p}:=( \int_X |f|^p dx)^{1/p}$ for $p< \infty$ and $\|f\|_{L^{\infty}}:=\mathrm{ess\,sup}|f|$.
    \end{enumerate}
    \item Distributions
    \begin{enumerate}
        \item The space of distributions $\mathcal{D}'(\R^D)$ is the dual space of $\mathcal{D}(\R^D)$.
        
        \item The space of tempered distributions $\mathcal{S}'(\R^D)$ is the dual space of $\mathcal{S}(\R^D)$.
    \end{enumerate}
    
    \item Complex Analysis
    \begin{enumerate}
        \item A function $f: \C^n \to \C$ is \textit{holomorphic} at a point $(z_1,\dots,z_n)\equiv\vec z$ if all of its first order complex derivatives exist as complex linear maps at $\vec z$.  
        Complex linearity requires $u D_{u \vec z} f=D_{\vec z} f$ for complex $u$, which implies $\bar\partial_{\vec z} f=0$.
        
        \item A function $f: \C^n \to \C$ is \textit{analytic} at a point $(z_1,\dots,z_n)$ if it can be represented by a convergent power series in a neighborhood of the point. One can prove that holomorphic functions are analytic (Goursat's theorem).

        \item A function $f: \C^n \to \C$ is \textit{entire} if it is analytic at every point in $\C^n$.
        
    \end{enumerate}
\end{enumerate}

\section{Proofs of Glimm-Jaffe Axioms} \label{app:osproof}

We mostly followed the proof structure laid out in Glimm and Jaffe with slight variations.

\subsection{Analyticity (OS0)}

The strict definition of analyticity of the GFF generating functional requires that 
\begin{equation}
    S\left[\sum_i z_i f_i\right] = \int\exp\left(i \sum_i z_i\, \phi(f_i)\right) d\mu(\phi)
\end{equation}
admits a converging power series expansion with respect to all $z_i$ for each fixed $\omega$. In practice, it would be hard to construct an expansion of $S\left[\sum_i z_i\, f_i\right]$ and prove its convergence; instead, we used Hartog's theorem to simplify the proof. 

Hartog's theorem states that in finite dimension, holomorphy of a complex function ($\mathbb{C}$-differentiable) implies its analyticity. Therefore OS0 simplifies to proving 
\begin{equation} \label{eq:OS0derv}
    \frac{\partial}{\partial z_j} \int\exp\left(i \sum_i z_i \, \phi(f_i)\right) d\mu(\phi)= \int i \, \phi(f_j) \cdot \exp\left(i \sum_i z_i \, \phi(f_i)\right)d\mu(\phi)
\end{equation}
is well-defined. This requires that \\
1. the integrand $\exp\left(i \sum_i z_i \,\phi(f_i)\right)$ is measurable, integrable (for $S\left[\sum_i z_i\, f_i\right]$ to be well-defined), and analytic (derivative exists);\\
2. the $z$ derivative is bounded by an integrable function $B(\phi)$, so that the derivative and integral can be exchanged;\\
3. the $z$ derivative is measurable, and hence integrable when combined with (2).

The hardest part in the proof was to establish a bound on the derivative. As we can see on the RHS of \ref{eq:OS0derv}, its norm is bounded by
\begin{equation}
    | i \, \phi(f_j) \cdot \exp\left(i \sum_i z_i \, \phi(f_i)\right)|\leq \sum_j |\, \phi(f_j)|\exp\left(-\, [\phi(f_i)]_{\text{im}}\right).
\end{equation}
To prove integrability of the bound, we use Fernique's Theorem, which states that for any $f$, there exists $\alpha>0$ such that
\begin{equation}
    \int \exp\left(\alpha \,[ \phi(f)]^2\right) d\mu(\phi) < \infty,
\end{equation}
and hence we can show $|\, \phi(f_j)|,\exp\left(-[\phi(f)]_{\text{im}}\right)\in L^2$. Lastly, we can apply Hölder's inequality, which says that product of $L^2$ functions is in $L^1$, to show that the RHS of \ref{eq:OS0derv} is integrable. In the actual proof, the exact bound is slightly different because we require the bound to be uniform in a neighborhood of $z_{0}$ (when taking the $\partial z_j$ derivatives), so we need to expand the bound function around $z_0$; still, the general logic follows the same way.

\subsection{Regularity (OS1)}

For GFF, we will prove \ref{eq:OS1} with $p=2$ and $c=\frac{1}{2m^2}$. The local integrability follows from the close-form expression of two-point functions:
\begin{equation} \label{eq:cov}
    K(x,y)=S_2(x,y)=\frac{m}{4\pi^2 |x-y|} K_1(m|x-y|).
\end{equation}
Let $y=0$ and take $x\to 0$, we have
\begin{equation} 
    \frac{m}{4\pi^2 |x|} K_1(m|x|)\sim \frac{1}{4\pi^2|x|^2}
\end{equation}
which is locally integrable in 4-dimension. 

First, by Minlos construction of the GFF, we have
\begin{equation}
    S(f)=\langle f,Kf\rangle\equiv\int d^4x\,d^4y\; f^*(x)\,K(x,y)\,f(y).
\end{equation}
To obtain an exponential bound, first notice that
\begin{equation}
    |S(f)|=\bigg|\exp\left(-\frac{1}{2}\langle f,Kf\rangle_{\mathbb{C}}\right)\bigg|=\exp\left(-\frac{1}{2}\text{Re}\langle f,Kf\rangle_{\mathbb{C}}\right),
\end{equation}
then decompose $f$ into its real and imaginary components
\begin{equation}
    \langle f, Kf \rangle_{\mathbb{C}} = \langle f_{\text{re}}, Kf_{\text{re}} \rangle + i \langle f_{\text{re}}, Kf_{\text{im}} \rangle + i \langle f_{\text{im}}, Kf_{\text{re}} \rangle - \langle f_{\text{im}}, Kf_{\text{im}} \rangle.
\end{equation}
Taking the real part of both sides we have
\begin{equation}
    \text{Re}\langle f, Kf \rangle_{\mathbb{C}} = \langle f_{\text{re}}, Kf_{\text{re}} \rangle - \langle f_{\text{im}}, Kf_{\text{im}} \rangle
\end{equation}

Since $K$ is positive semi-definite, $\langle f_{\text{re}}, Kf_{\text{re}} \rangle \geq 0$, so
\begin{equation}
    -\text{Re}\langle f, Kf \rangle_{\mathbb{C}} \leq \langle f_{\text{im}}, Kf_{\text{im}} \rangle.
\end{equation}
To bound $\langle f,Kf\rangle$, we can expand it in momentum space
\begin{equation}
    \langle f, Kf \rangle = \int \frac{|\hat{f}(k)|^2}{ |k|^2 + m^2} \, dk
\end{equation}
where $\hat{f}(k)$ is the fourier transform of $f$. Notice that the denominator is bounded by
\begin{equation}
    \frac{1}{|k|^2+m^2}\leq\frac{1}{m^2}
\end{equation}
and by Plancherel's theorem ($\int |f(x)|^2 dx = \int |\hat{f}(k)|^2dk$), we get the bound
\begin{equation}
    \langle f, Kf \rangle\leq \frac{1}{m^2}\|f\|^2_{L^2},
\end{equation}
and therefore
\begin{equation}
    \langle f_{\text{im}}, Kf_{\text{im}} \rangle\leq \frac{1}{m^2}\|f_{\text{im}}\|^2_{L^2}\leq\frac{1}{m^2}\|f\|^2_{L^2}
\end{equation}
and
\begin{equation}
    |S(f)| = \exp\left(-\frac{1}{2}\text{Re}\langle f, Kf \rangle\right) \leq \exp\left(\frac{1}{2}\langle f_{\text{im}}, Kf_{\text{im}} \rangle\right) \leq \exp\left(\frac{1}{2m^2} \|f\|_{L^2}^2\right)\leq \exp\left(\frac{1}{2m^2} (\|f\|_{L^1}+\|f\|_{L^2}^2)\right).
\end{equation}

\subsection{Euclidean Invariance (OS2)}

We would like to show that under any Euclidean transformation $E$ (translation and rotation), $S(f)$ remains unchanged. First we perform a change of variable by $x\to E^{-1}x$ and $y\to E^{-1}y$: 
\begin{equation}
    S(Ef)=\int d^4x\,d^4y\; f^*(Ex)K(x,y)f(Ey)=\int d^4(E^{-1}x)\,d^4(E^{-1}y)\; f^*(x)K(E^{-1}x,E^{-1}y)f(y).
\end{equation}
From \ref{eq:cov}, we can see that $K(x,y)$ only depends $|x-y|$, and therefore $K(E^{-1}x,E^{-1}y)=K(x,y)$ is invariant under Euclidean group action. Furthermore, the Lebesgue measure is also invariant under translation and rotation ($|\det R|=1$), so $ d^4(E^{-1}x)=d^4x$. Combining these facts, we can conclude that 
\begin{equation}
    S(Ef)=\int d^4(x)\,d^4(y)\; f^*(x)K(x,y)f(y)=S(f).
\end{equation}

\subsection{Reflection Positivity (OS3)}

The statement of OS3 as in \ref{eq:OS3} refers to the reflection positivity of the generating functional. First, we would want to use the Schur product theorem to reduce this statement to the reflection positivity of the free covariance. An equivalent statement with \ref{eq:OS3} is that for every finite sequence of functions $f_{j}$, $S\{f_i-\theta f_j\}$ is positive semi-definite. We want to show that 
\begin{equation}
    \langle \theta f,Kf\rangle\geq 0 \Rightarrow S\{f_i-\theta f_j\}=\exp\left(-\frac{1}{2}\langle f_i - \Theta f_j, K(f_i - \Theta f_j) \rangle\right)\succeq 0
\end{equation}
for any $f$ defined on positive time. Rewriting the exponent on the RHS gives 
\begin{equation}
    \langle f_i - \Theta f_j, K(f_i - \Theta f_j) \rangle = \langle f_i, Kf_i \rangle + \langle f_j, Kf_j \rangle - 2\langle \Theta f_i, Kf_j \rangle
\end{equation}
where we used the time-reflection symmetry of the inner product. Then we can expand the exponential
\begin{equation} \label{eq:OS3expand}
    \exp\left(-\frac{1}{2}\langle f_i - \Theta f_j, K(f_i - \Theta f_j) \rangle\right)=\exp\left(-\frac{1}{2}\langle f_i,Kf_i\rangle\right)\times\exp\left(-\frac{1}{2}\langle f_j,Kf_j\rangle\right)\times\exp\left(R_{ij}\right)
\end{equation}
where  $R_{ij}\equiv \langle \theta f_i,Kf_j\rangle$. For any coefficient $c$, the quadratic form is
\begin{equation}
    \sum_{i,j} c_i c_j R_{ij} = \left\langle \Theta\left(\sum_i c_i f_i\right), K\left(\sum_j c_j f_j\right) \right\rangle\geq0
\end{equation}
where the last inequality is given by free covariance reflection positivity, treating $\sum c_i f_i$ as one function. Since $R_{ij}$ is positive semi-definite, by Schur product theorem (product of PSD is also PSD), we conclude that $\exp\left(R_{ij}\right)$ is PSD. Finally, looking at the quadratic form of \ref{eq:OS3expand}, we have
\begin{equation}
    \sum_{i,j}c_ic_jS(f_i-\theta f_j)=\sum_{i,j}\bigg[c_i\exp\left(-\frac{1}{2}\langle f_i,Kf_i\rangle\right)\bigg]\bigg[c_j\exp\left(-\frac{1}{2}\langle f_j,Kf_j\rangle\right)\bigg]\exp\left(R_{ij}\right)\geq 0
\end{equation}
which is exactly the condition for OS3.

Now the proof os OS3 simplifies to proving the reflection positivity of the free covariance $\langle \theta f,Kf\rangle$. Expanding the inner product gives
\begin{equation}
\begin{split}
        \langle \theta f,K f\rangle 
        &\equiv \int d^4x\, d^4y\,f^*(-x_0, \vec{x}) f(y_0, \vec{y})K(x,y)\\
    &=\int d^4x\, d^4y\,f^*(x_0, \vec{x}) f(y_0, \vec{y})\int d^3k\,dk_0\, \frac{1}{k^2+m^2}e^{-ik_0(x_0+y_0)+i\vec{k}\cdot\vec{x}},
\end{split}
\end{equation}
where $\theta$ is the time reflection operator and $K(x,y)$ is the free covariance, or the Green's function. Naively, we want to first perform the $k_0$ integral to get
\begin{equation}
    \int k_0\, \frac{1}{k^2+m^2}e^{-ik_0(x_0+y_0)} = \frac{\pi}{\omega_k}e^{-\omega_k |x_0+y_0|},
\end{equation}
where $\omega_k=\sqrt{\vec{k}^2+m^2}$. Then we commute the $\vec{k}$ integral outside the $x$ and $y$ integrals, and perform the $\vec{x}$ and $\vec{y}$ integrals to get
\begin{equation}\label{eq:naiveOS3}
\begin{split}
    &\int d^4x\, d^4y\int d^3k \,f^*(x_0, \vec{x}) f(y_0, \vec{y})\frac{\pi}{\omega_k}e^{-\omega_k |x_0+y_0|}e^{i\vec{k}\cdot\vec{x}}\\
    &=\int d^3k \int d^4x\, d^4y\,f^*(x_0, \vec{x}) f(y_0, \vec{y})\frac{\pi}{\omega_k}e^{-\omega_p |x_0+y_0|}e^{i\vec{k}\cdot\vec{x}}\\
    &=\int d^3k\int dx_0\, dy_0 \,\frac{\pi}{\omega_p}\tilde{f^*}(x_0,\vec{k})\tilde{f}(y_0,\vec{k})e^{-\omega_k |x_0+y_0|}\\
    &=\int d^3k\,\frac{\pi}{\omega_k} \bigg| \int dt\, f(t,\vec{k})e^{-\omega_p t} \bigg|^2 \geq 0
\end{split}
\end{equation}
given that $\omega_k\geq0$ and $f$ is supported in positive time only. However, this did not work so simply in lean, as we require the integrand to be absolutely convergent to use Fubini theorem for exchanging order of integration. The integrand in the first line of \ref{eq:naiveOS3} is NOT absolutely integrable - the integrand is finite for $x\neq y$, but it relies on the oscillating property. If we take the absolute value, the factor $1/\sqrt{k^2+m^2}$ is divergent under the 3D $k$-integral. To resolve this, we need to make use of the Schwinger parametrization:
\begin{equation}\label{eq:heatExpand}
\begin{split}
    \langle \theta f,K f\rangle 
        &\equiv \int d^4x\, d^4y\,f^*(x) f(y)K(\theta x,y)\\
        &=\int_x \int_y f^*(x)\, f(y) \left[\int_0^\infty e^{-sm^2}\, H(s, |\Theta x - y|)\, ds\right] \\
        &= \int_0^\infty e^{-sm^2} \left[\int_x \int_y f^*(x)\, f(y)\, H(s, |\Theta x - y|)\right] ds
\end{split}
\end{equation}
where 
\begin{equation}
    H(s,r)\equiv \frac{1}{(4\pi s)^2}\exp\left(-\frac{r^2}{4s}\right)
\end{equation}
is the heat kernel, and bounded by $s^{-2}$. The integrand is absolute integrable in this case: the Schwartz functions $f^*(x)$ and $f(y)$ are bounded and integrable in $x$ and $y$ by definition; $H(s,r)$ is bounded, in addition to the exponential decay $e^{-sm^2}$ on $s$.\footnote{The exact proof can be found in OS3\_foundations, theorem named Schwinger\_bilinear\_integrable.} Therefore, we are allowed to perform the change of integration order in \ref{eq:heatExpand}.

We then use the fourier representation of the heat kernel
\begin{equation}
    H(s,|z|) = (2\pi)^{-4} \int_k \exp(-ik\cdot z) \exp(-s|k|^2)
\end{equation}
to rewrite the full integral as
\begin{equation}\label{eq:fullOS3}
    \int_{s\geq0} \int_x \int_y\int_{\vec{k}}\int_{k_0} e^{-sm^2}  f^*(x)\, f(y)\, \exp(-ik\cdot (\theta x-y)) \exp(-s|k|^2).
\end{equation}
The $k_0$ integral yields
\begin{equation}
    \int e^{ik_0(x_0+y_0)}e^{-sk_0^2}\,dk_0=\sqrt{\frac{\pi}{s}}e^{-(x_0+y_0)^2/4s}.
\end{equation}
In order to apply Fubini theorem, we will show that the absolute value of the integrand
\begin{equation}\label{eq:schwingerInt}
    \bigg|\sqrt{\frac{\pi}{s}}e^{-(x_0+y_0)^2/4s}e^{-s(m^2+\vec{k}^2)}  f^*(x)\, f(y)\, e^{i\vec{k}\cdot\vec{x}} \bigg|
\end{equation}
is bounded by an integrable function. When taking the norm, $\exp(i\vec{k}\cdot\vec{x})$ has norm 1 and drops out. We can bound the Schwartz functions by
\begin{equation}
    |f(x)||f(y)|\leq C\cdot\frac{x_0 y_0}{(1+|\vec{x}|^2)^N(1+|\vec{y}|^2)^N} 
\end{equation}
for some constant $C$ and any $N\in\mathbb{Z}^+$, given that $x_0,y_0>0$ as $f$ is only defined on positive time. Then we can use the Gaussian moment formula on the $x_0$ and $y_0$ integrals, and use $N=4$ for the spatial integrals to get
\begin{equation}
\begin{split}
    &\int_x\int_y\,\bigg|\sqrt{\frac{\pi}{s}}e^{-(x_0+y_0)^2/4s}e^{-s(m^2+\vec{k}^2)}  f^*(x)\, f(y)\, e^{i\vec{k}\cdot\vec{x}} \bigg|\\
    &\leq C \int_x\int_y\sqrt{\frac{\pi}{s}}\frac{x_0 y_0}{(1+|\vec{x}|^2)^4(1+|\vec{y}|^2)^4}e^{-(x_0+y_0)^2/4s}e^{-s(m^2+\vec{k}^2)}\\
    &=\frac{4C'}{3}\sqrt{\frac{\pi}{s}}e^{-s(m^2+\vec{k}^2)}\sqrt{\pi}s^{2}
\end{split}
\end{equation}
where $C'$ is $C$ multiplied by the finite result from the spatial integrals. This shows that \ref{eq:schwingerInt} is integrable under the $(x,y)$ double integration. Then we can prove that this is still integrable under the $\vec{k}$ ans $s$ integrals:
\begin{equation}
    \int_{s\geq0}\int_{\vec{k}}e^{-s(m^2+\vec{k}^2)}s^{3/2}=\int_{s\geq 0}e^{-sm^2}s^{3/2}(\pi/s)^{3/2}=\frac{\pi^{3/2}}{m^2}
\end{equation}
is finite. 

After rearranging the order of integration in \ref{eq:fullOS3}, we have
\begin{equation}
\begin{split}
    &\int_{\vec{k}}\int_x\int_y\int_{s\geq0}\,\sqrt{\frac{\pi}{s}}e^{-(x_0+y_0)^2/4s}e^{-s(m^2+\vec{k}^2)}  f^*(x)\, f(y)\, e^{i\vec{k}\cdot\vec{x}} \\
    &=\int_{\vec{k}}\int_x\int_y\,\frac{\pi}{\omega}e^{-\omega|x_0+y_0|}f^*(x)\, f(y)\, e^{i\vec{k}\cdot\vec{x}}
\end{split}
\end{equation}
where we used $\int_0^\infty s^{-1/2}\exp(-a/s - bs) \,ds = \sqrt{\pi/b} \exp(-2\sqrt{ab})$. This is essentially the exchange of integration order we wanted in \ref{eq:naiveOS3}, and the positivity simply follows as proposed.

\subsection{Ergodicity (OS4)}\label{subsection:OS4}

Once we know the polynomial clustering statement \eqref{eq:OS4-Poly-Clustering}, OS4 follows quite readily. First instead of proving ergodicity for functions of the form $A(\phi)=\sum_j c_j e^{(\phi,f_j)}$, we can reduce the claim to working with single exponential functions $ e^{(\phi,f)}$ by Minkowski's ($L^2$ triangle) inequality and linearity of the integral.
\begin{equation}
    \bigg\|\frac{1}{t}\int_0^t A(T_s \phi) ds-\E_{\mu}[A(\phi)]\bigg\|_{L^2(\mu)} \leq \sum_{j}|c_j| \bigg\|\frac{1}{t}\int_0^t e^{(T_s\phi,f_j)} ds-\E_{\mu}[e^{(\phi,f_j)}]\bigg\|_{L^2(\mu)}.
\end{equation}
As a note for implementation in lean, even this seemingly trivial step was not so easy to implement when presented to Claude code. In particular the difficulty was in verifying  measurability for basic operations such as linearity of the integral. It was useful to tell Claude code to make general measurability lemmas which would apply to later arguments.

The next step is to use linearity and Fubini as follows,
    \begin{equation}\label{eq:L^2-comp}
        \begin{split}
        \bigg\| \E_{\mu}[e^{(\phi,f)}]-\frac{1}{t} \int_{0}^{t} e^{(T_s\phi,f)}ds\bigg\|_{L^2}^2&= \bigg\| \frac{1}{t}\int_{0}^{t} (\E_{\mu}[e^{(\phi,f)}]- e^{(T_s\phi,f)}) ds\bigg\|_{L^2}^2\\ 
        &= \frac{1}{t^2}\int_{0}^{t} \int_{0}^{t} \E_{\mu}[(\E_{\mu}[e^{(\phi,f)}]- e^{(T_s\phi,f)})(\overline{\E_{\mu}[e^{(\phi,f)}]- e^{(T_{s'}\phi,f)}})]ds ds'.
        \end{split}
    \end{equation}
    The justification for Fubini follows from the Fernique bound which was once again the most difficult part of the above equation display to implement in lean.
    
    After a little algebra, using the easy to prove property that $C(T_s g,T_s g)=C(g,g)$ for any Schwartz function $g$ (this is proven as a lemma in our lean repository), and
    plugging in the polynomial clustering bound,
\begin{equation}
    \begin{split}
        ~&|\E_{\mu}[(e^{(\phi,f)}- e^{(T_s\phi,f)})(e^{(\phi,f)}- e^{(T_{s'}\phi,f)})]| \\
        &=|\E_{\mu}[e^{(\phi,f)}]\E_{\mu}[e^{(\phi,\overline{f})}]-\E_{\mu}[e^{(T_s\phi,f)+(T_{s'}\phi,\overline{f})}]|\\
        & \leq c(1+|s-s'|)^{-2}.
        \end{split}
    \end{equation}
    Plugging this bound back into \eqref{eq:L^2-comp} we have,
    \begin{equation}\label{eq:OS4'-to-OS4''-last-comp}
        \begin{split}
            \bigg\| \E_{\mu}[e^{(\phi,f)}]-\frac{1}{t} \int_{0}^{t} e^{(T_s\phi,f)}ds\bigg\|_{L^2} &\leq c \frac{1}{t^2}\int_{0}^{t}\int_{0}^{t} (1+|s-s'|)^{-2} ds ds' \\
        &\leq c\sup_{s' \in [0,t]}\frac{1}{t}\int_{0}^{t} (1+|s-s'|)^{-2} ds\\
        &\leq c\frac{1}{t}\int_{-\infty}^{\infty}  (1+|s|)^{-2}ds \to 0
        \end{split}
    \end{equation}
    as $t \to\infty$. The equation display above was fairly simple to implement in lean. \\

The remaining task for OS4 is now to prove the polynomial clustering version \eqref{eq:OS4-Poly-Clustering}. This property follows from the following two key results

\begin{enumerate}
    \item (Pointwise decay of correlation functions) The covariance function $K(x,y)$ satisfies a bound of the form,
    \begin{equation}\label{eq:pointwise-decay-of-correlations}
        |K(x,y)| \leq c|x-y|^{-2}e^{-m|x-y|}
    \end{equation}
    for some constant $c>0$.
    \item (Decay of averaged correlation functions) Let $E$ be a finite-dimensional real normed vector space with $d=\dim E$.
    Let $f,g\in\mathcal{S}(E;\mathbb{C})$ be Schwartz functions.
    Let $K:E\to\mathbb{R}$ be measurable and locally integrable. Assume there exist
    $m>0$, $C_K>0$, $R_0>0$ such that for all $z\in E$ with $\|z\|\ge R_0$,
    \[
    |K(z)|\le C_K e^{-m\|z\|}.
    \]
    Then for every $\alpha>0$ there exists $c\ge 0$ such that for all $a\in E$,
    \begin{equation}\label{eq:schwartz-decay}
    \Bigl\|\int_{x\in E}\int_{y\in E} f(x)\,K(x-y)\,g(y-a)\,dy\,dx\Bigr\|
    \le c\,(1+\|a\|)^{-\alpha}.
    \end{equation}
\end{enumerate}

The bound on the covariance above follows from standard large and small scale asymptotics of the Bessel function $K_1(z)$ which we were able to implement in lean.

We give more details for the proof of \eqref{eq:schwartz-decay} since the fact we proved it for Schwartz functions required some additional argument which one would not find in the textbook of Glimm and Jaffe. The first step is the following lemma.
\begin{lemma}[Convolution polynomial decay]
Assume $N>d$. Let $F,G:E\to\mathbb{C}$ be integrable and satisfy
\[
|F(y)|\le \frac{C_F}{(1+\|y\|)^N},\qquad |G(y)|\le \frac{C_G}{(1+\|y\|)^N}.
\]
Define $h(x):=\int_E F(y)\,G(x-y)\,dy$. Then for all $x\in E$,
\[
|h(x)|\le \frac{C}{(1+\|x\|)^N}
\]
for $C := 2^N\bigl(C_F\|G\|_{L^1}+C_G\|F\|_{L^1}\bigr)$.
\end{lemma}

\emph{Proof (domain split at $\|y\|=\|x\|/2$).}
Fix $x\in E$ and split $E=A\cup B$ with
\[
A:=\{y:\|y\|\ge \|x\|/2\},\qquad B:=\{y:\|y\|< \|x\|/2\}.
\]
Then
\[
|h(x)|\le \int_A |F(y)|\,|G(x-y)|\,dy + \int_B |F(y)|\,|G(x-y)|\,dy.
\]
On $A$ we have $1+\|y\|\ge (1+\|x\|)/2$, hence $(1+\|y\|)^{-N}\le 2^N(1+\|x\|)^{-N}$ and
\[
|F(y)|\le \frac{C_F}{(1+\|y\|)^N} \le \frac{2^N C_F}{(1+\|x\|)^N}.
\]
Therefore
\[
\int_A |F(y)|\,|G(x-y)|\,dy
\le \frac{2^N C_F}{(1+\|x\|)^N}\int_E |G(x-y)|\,dy
= \frac{2^N C_F}{(1+\|x\|)^N}\|G\|_{L^1},
\]
using translation invariance in the last equality.

On $B$, the triangle inequality gives $\|x-y\|\ge \|x\|-\|y\|>\|x\|/2$, hence similarly
\[
|G(x-y)|\le \frac{2^N C_G}{(1+\|x\|)^N}.
\]
Thus
\[
\int_B |F(y)|\,|G(x-y)|\,dy
\le \frac{2^N C_G}{(1+\|x\|)^N}\int_E |F(y)|\,dy
= \frac{2^N C_G}{(1+\|x\|)^N}\|F\|_{L^1}.
\]
Adding the two bounds gives the claimed estimate. \qed

The rest of the proof of \eqref{eq:schwartz-decay} involves similar arguments decomposing the domain of integration into `singular' and `tail' regions so we assume it for the rest of the argument.

\begin{proof}[Proof sketch of OS4 (polynomial clustering) for the massive GFF]
Fix $\alpha>0$ and test functions $f,g\in\mathcal{S}(\mathbb{R}^4;\mathbb{C})$. By definition of $T_s\omega$ as the dual action on distributions,
and bilinearity of the inner product,
\[
\mathbb{E}_\mu\bigl[e^{\langle \omega,f\rangle + \langle T_s\omega,g\rangle}\bigr]
=\mathbb{E}_\mu\bigl[e^{\langle \omega,f\rangle + \langle \omega,T_{-s}g\rangle}\bigr]
=\mathbb{E}_\mu\bigl[e^{\langle \omega,f+T_{-s}g\rangle}\bigr].
\]

Next, for a centered Gaussian field with covariance bilinear form $S_2(\cdot,\cdot)$ (the Schwinger 2-point
function), one has the standard identity
\[
\mathbb{E}_\mu\bigl[e^{\langle \omega,f+T_{-s}g\rangle}\bigr]
= \mathbb{E}_\mu\bigl[e^{\langle \omega,f\rangle}\bigr]\,
  \mathbb{E}_\mu\bigl[e^{\langle \omega,T_{-s}g\rangle}\bigr]\,
  \exp\!\bigl(S_2(f,T_{-s}g)\bigr).
\]

Moreover, by time-translation invariance of the GFF, $\mathbb{E}_\mu\bigl[e^{\langle \omega,T_{-s}g\rangle}\bigr]
=\mathbb{E}_\mu\bigl[e^{\langle \omega,g\rangle}\bigr]$. So defining the constants,
$E_f:=\mathbb{E}_\mu\bigl[e^{\langle \omega,f\rangle}\bigr],\qquad
E_g:=\mathbb{E}_\mu\bigl[e^{\langle \omega,g\rangle}\bigr]$, we have,
\begin{equation}\label{eq:diff}
\mathbb{E}_\mu\bigl[e^{\langle \omega,f\rangle + \langle T_s\omega,g\rangle}\bigr]-E_fE_g
=E_fE_g\bigl(\exp(S_2(f,T_{-s}g))-1\bigr),
\end{equation}
and using the elementary inequality (valid for all $z\in\mathbb{C}$), $|\exp(z)-1|\le |z|\,e^{|z|}$,
\begin{equation}\label{eq:expminus1}
\Bigl|\mathbb{E}_\mu\bigl[e^{\langle \omega,f\rangle + \langle T_s\omega,g\rangle}\bigr]-E_fE_g\Bigr|
\le |E_f|\,|E_g|\,|S_2(f,T_{-s}g)|\,\exp(|S_2(f,T_{-s}g)|).
\end{equation}
So it suffices to prove a polynomial bound
\begin{equation}\label{eq:S2poly}
|S_2(f,T_{-s}g)|\le c_{\mathrm{decay}}\,(1+s)^{-\alpha}\qquad (s\ge 0),
\end{equation}
with some constant $c_{\mathrm{decay}}\ge 0$ depending on $(f,g,\alpha)$.

\medskip
Finally recalling that
\begin{equation}\label{eq:S2asdouble}
S_2(f,T_{-s}g)=\int_x\int_y f(x)\,K(x-y)\,g(y-\tau_s)\,dy\,dx.,
\end{equation}
we are now in a position to apply the main Schwartz double integral decay bound \eqref{eq:schwartz-decay}.
As a result, there exists a constant $c_{\mathrm{decay}}\ge 0$ such that for all $a\in\mathbb{R}^4$,
\[
\Bigl|\int_x\int_y f(x)\,K(x-y)\,g(y-a)\,dy\,dx\Bigr|
\le c_{\mathrm{decay}}\,(1+\|a\|)^{-\alpha}.
\]
Taking $a=\tau_s$ and using $1+\|\tau_s\|=1+s$ gives exactly \eqref{eq:S2poly}.

Finally combining \eqref{eq:expminus1} and \eqref{eq:S2poly} we have
\[
\Bigl|\mathbb{E}_\mu\bigl[e^{\langle \omega,f\rangle + \langle T_s\omega,g\rangle}\bigr]-E_fE_g\Bigr|
\le |E_f|\,|E_g|\,c_{\mathrm{decay}}(1+s)^{-\alpha}\,
\exp\!\bigl(c_{\mathrm{decay}}(1+s)^{-\alpha}\bigr),
\]
and since $(1+s)^{-\alpha}\le 1$, \eqref{eq:OS4-Poly-Clustering} holds with constant
\[
 |E_f\|\,\|E_g\|\,c_{\mathrm{decay}}\,e^{c_{\mathrm{decay}}}\ge 0,
\]
which is the desired OS4 polynomial clustering bound.
\end{proof}

Finally, for completeness the file {\tt OS/NonTrivial} verifies that the correlation functions are
nontrivial, as taking $\phi(x)=0$ would satisfy the axioms.

\end{document}